\documentclass[a4paper,fleqn,11pt]{cas-sc}

\usepackage[authoryear,longnamesfirst]{natbib}

\def\tsc#1{\csdef{#1}{\textsc{\lowercase{#1}}\xspace}}
\tsc{WGM}
\tsc{QE}

\usepackage{amssymb} 
\usepackage{amsmath}
\usepackage{tikz}
\usepackage{amsthm}
\usepackage{mathtools}
\usepackage{physics}
\usepackage{derivative}

\newcount\Comments  
\newcommand{\kibitz}[2]{\ifnum\Comments=1{\color{#1}{#2}}\fi}

\usetikzlibrary{positioning}
\usetikzlibrary{intersections}
\usetikzlibrary{calc}
\usepackage{tkz-euclide}
\usetikzlibrary{arrows}
\input{arrowsnew} 
\usepackage[normalem]{ulem}

\usepackage{graphicx}
\usepackage{caption}
\usepackage{subfigure}

\usepackage{xparse}

\NewDocumentCommand{\twohalfcircle}{ O{#2} m }{%
    \begin{tikzpicture}
    \fill[#2] (0,0) -- (90:0.5ex) arc (90:270:0.5ex) -- cycle;; 
    \fill[#1] (0,0) -- (90:0.5ex) arc (90:-90:0.5ex) -- cycle; 
    \draw (0,0) circle (0.5ex);
    \end{tikzpicture}
}

\renewcommand{\H}{  
    \textbf{H}
}

\newcommand{\B}{  
    \textbf{B}
}

\newcommand{\He}{  
    \B^{\twohalfcircle[black]{white}}
}

\usepackage{colortbl}

\newcolumntype{x}[1]{>{\centering\arraybackslash\hspace{0pt}}p{#1}}

\usepackage{mathtools}
\usepackage{tikz}
\usepackage{tikz-3dplot}
\usetikzlibrary{arrows,calc,backgrounds}
\usetikzlibrary{decorations.markings}
\usepackage{bbm}

\pgfmathsetmacro{\r}{3} %
\pgfmathsetmacro{\ra}{2}
\pgfmathsetmacro{\h}{21/8} %
\tdplotsetmaincoords{20}{20}

\newtheorem{theorem}{Theorem}
\newtheorem{lemma}[theorem]{Lemma}

\newtheorem{proposition}[theorem]{Proposition}
\newtheorem{definition}{Definition}

\newproof{pf}{Proof}

\begin{document}
\let\WriteBookmarks\relax
\def\floatpagepagefraction{1}
\def\textpagefraction{.001}

\shorttitle{On the Approximability of the Yolk in the Spatial Model of Votings}    

\shortauthors{}  

\title [mode = title]{On the Approximability of the Yolk in the Spatial Model of Voting}  



%

\author{Ran Hu}[orcid=0000-0003-3133-6511]

\ead[url]{https://ise.rpi.edu/people/ran-hu}

\ead{hur6@rpi.edu}

\author{James P. Bailey}[orcid=0000-0002-7207-512X]



\ead{bailej6@rpi.edu}

\ead[url]{www.jamespbailey.com}


\affiliation{organization={Rensselaer Polytechnic Institute (RPI)},
            addressline={110 8th St}, 
            city={Troy},
            postcode={12180}, 
            state={NY},
            country={USA}}




\begin{abstract}
    In the spatial model of voting, the yolk and LP (linear programming) yolk are important solution concepts for predicting outcomes for a committee of voters. McKelvey and Tovey showed that the LP yolk provides a lower bound approximation for the size of the yolk and there has been considerable debate on whether the LP yolk is a good approximation of the yolk.
    In this paper, we show that for an odd number of voters in a two-dimensional space that the yolk radius is at most twice the size of the LP yolk radius. 
    However, we also show that (1) even in this setting, the LP yolk center can be arbitrarily far away from the yolk center (relative to the radius of the yolk) and (2) for all other settings (an even number of voters or in dimension $k\geq 3$) that the LP yolk can be arbitrarily small relative to the yolk. 
    Thus, in general, the LP yolk can be an arbitrarily poor approximation of the yolk.
\end{abstract}



\maketitle


\section{Introduction}
The spatial model of voting \citep{enelow1984spatial,miller2015spatial} is a theoretical framework used in political science and economics to analyze voting behaviors (see e.g., \cite{reynolds1969spatial,merrill1993voting}), candidate strategies, and electoral outcomes \citep{kenny2005evidence}.
It models preferences and choices in a multidimensional space, where each dimension in the space represents a different political issue or ideology, e.g., economic policy on one axis and social policy on another.
In this model, voters are represented as \textit{ideal} points within this multidimensional space. 
A voter's position in this space represents their ideal point or preferred position on various issues. 
In the standard spatial model of voting, voters preferences are single-peaked and they prefer policies candidates closer to their ideal point \citep{black1948rationale}.

The spatial voting model helps explain how individuals make voting decisions based on their policy preferences and the positions of different candidates (see e.g., \cite{boudreau2015lost}).
It also can be used to predict election outcomes by identifying where candidates stand on the political spectrum in relation to the electorate's distribution of preferences.
Understanding the spatial distribution of voter preferences allows political candidates and parties to tailor their campaigns more effectively \citep{owen1989optimal}.
By clarifying how democratic processes work and how policy decisions are made, spatial voting theory can contribute to a better understanding of democracy. 
\cite{jenkins1998spatial} extend research on the spatial theory of voting and the electoral relationship, using a spatial model of voting as a baseline to test the corrupt bargaining hypothesis.
\cite{endersby1992stability} compare candidate positions across time by using the spatial theory of voting.
In \cite{rosenthal1977spatial}, the formal spatial models are applied to cross-sectional analysis of district results on the second ballot of French legislative elections, and the results prove the importance of spatial models in empirical analysis.
\cite{ccarkouglu2006spatial} adopt the theoretical framework of the spatial voting model to conduct a multidimensional analysis of the emerging ideological space in Turkey in the years leading up to the November 2002 elections.


In spatial voting theory, several solution concepts help understand the dynamics of voter preferences, candidate strategies, and the resulting equilibria.
The \textit{core} \citep{schofield1988core} in spatial voting models is a set of policy positions that cannot be defeated by any other position under majority rule. 
A position within the core is stable since no majority of voters prefers another position. 
Ever since \cite{feld1987necessary}, \cite{schofield1978instability} and \cite{mckelvey1979general} showed that an exact core typically does not exist, many solution concepts have been proposed.
The $\epsilon$\textit{-core} \citep{maschler1979geometric,wooders1983epsilon,tovey1993some,tovey1995dynamical} is proposed after allowing the stability conditions to be relaxed slightly. 
The $\epsilon$-core consists of policy positions such that no coalition of voters can deviate and make themselves at least $\epsilon$ better off by choosing another policy position, with $\epsilon$ being a positive but arbitrarily small amount. 
The \textit{yolk} is an important concept in spatial voting theory proposed by \cite{ferejohn1984limiting} and \cite{mckelvey1986covering}. 
Unlike the core, which is defined by policy stability under majority voting, the yolk is a geometric concept that tries to approximate the center of voter preferences in the absence of a clear core.
A yolk is a smallest ball that intersects every median hyperplane, a median hyperplane is a plane that divides the voter population into two equal halves. 
The center of a yolk is a point that, in some sense, represents a compromise position that minimizes the maximum distance to the voters' ideal points. 
It is used to predict the outcomes of majority rule voting in multidimensional policy spaces. 
An \textit{LP (linear programming) yolk} \citep{stone1992limiting,koehler1992limiting} is defined to be the ball with smallest radius that intersects every limiting median hyperplane \citep{miller1989geometry} -- in $\mathbb{R}^k$, a hyperplane is a limiting median hyperplane if it passes through at least $k$ of the ideal points.
The \textit{finagle point} is introduced by \cite{wuffle1989finagle}, it is the point with minimum finagle radius, where the finagle radius is the radius of a circle such that an alternative in the circle can defeat any alternative in the space if the candidate is located at its center. 
When the core is non-empty, all of these solution concepts are consistent, i.e., they all yield the core.
Relationships among these four concepts of core, $\epsilon$-core, yolk and finagle point are investigated in the work of \cite{brauninger2007stability}, \cite{nganmeni2023finagle} and \cite{martin2019dominance}. 
In this paper, we focus on the relationship between the yolk and LP yolk concepts.

Compared with yolk, LP yolk is easier to compute, and the radius is smaller than yolk. 
McKelvey's original introduction of the yolk includes a relatively simple, polynomial sized linear program to compute the LP yolk center and its radius \citep{mckelvey1986covering}.
\cite{Tovey92} presents a polynomial-time algorithm for computing the yolk in fixed dimension --- in two dimensions, it computes the yolk in $O(n^{4.5})$ time.
\cite{berg2018faster} provides an $O(n^{4/3} \log^{1+\epsilon} n)$ time algorithm to compute minimum-radius plurality ball (or alternatively, the yolk).
\cite{gudmundsson2019computing} present a near-linear time algorithms for approximately computing the yolk in the plane, by applying Megiddo’s parametric search technique, breaking the best known upper bound of $O(n^{4/3})$.

In this paper, we provide tight bounds on the approximability of the yolk via the LP yolk.
In general, the relationship between the yolk and LP yolk has an interesting history. 
McKelvey introduces the yolk --- equivalently, the generalized median set --- in \cite{mckelvey1986covering} as the smallest ball intersecting all median hyperplanes. 
It is then straightforward to show that the yolk can be found with via linear programming with an infinite number of constraints where each constraint requires the center of the yolk to be within radius $r$ of the corresponding median hyperplane. 
However, McKelvey claims in his formulation that it suffices to consider only the set of $O(n^k)$ limiting median hyperplanes (i.e., the LP yolk), thus yielding a finite linear program which is polynomial-time solvable. 
However, several years later, Stone and Tovey prove that the set of limiting median hyperplanes is insufficient to determine the yolk \citep{stone1992limiting}, which led to a series of papers comparing the yolk and LP yolk.
\cite{koehler1992limiting} experimentally shows that for a large number of voters that the set of median hyperplanes frequently are sufficient to determine the yolk. 
However, this is partially explained by two papers that show probabilistic convergence as the number of voters increase; Tovey shows that both radii converge to 0 \citep{Tovey2010} and McKelvey and Tovey show that the center and radius of the LP yolk converges to the center and radius of the yolk \citep{mckelvey2010approximation}.
Thus, for a large number of voters in low dimensions, the LP introduced by McKelvey provides a quick, close approximation of the yolk. 
However, many models of spatial voting are built around small committees in two dimensions.
The probabilistic approximation guarantees provided by \cite{Tovey2010} and the observations by \cite{koehler1992limiting} do not address this setting. 
Our paper aims to close this gap on the approximability of the yolk by providing tight bounds comparing the two concepts.

\subsection{Our Contributions}
In this paper, we show that the LP yolk radius is at least 1/2 the size of the yolk radius for an odd number of voters in $\mathbb{R}^2$ and demonstrate that the bound is asymptontically tight. 
We also show that for an odd number of voters in $\mathbb{R}^2$, the LP yolk center can be arbitrarily far away from the yolk center. 
Furthermore, when dimension $k\ge3$ or an even number of voters, we show that the LP yolk can be arbitrarily small relative to the yolk. 
Therefore, the LP yolk can generally be an poor approximation of the yolk.

\section{Notation}
In the spatial model of voting, a voter $v$ is represented by their ideal point $\pi_v$ in a space of dimension $k$. The set of ideal points is denoted by $I\subset \mathbb{R}^k$. 
\begin{definition}
  The set of all hyperplanes is given by ${\cal H} = \{(a, b) \in \mathbb{R}^k \times \mathbb{R}: \lVert a \rVert_2 = 1\}$. The hyperplane (a, b) refers to the set of points $\H(a, b) = \{x \in \mathbb{R}^k: a\cdot x = b \}$.
\end{definition}

Requiring $||a||_2=1$ loses no generality -- since the hyperplane $(a,b)$ is equivalent to the hyperplane $(a,b)/||a||_2$.  
Further, requiring $a$ to be a unit vector has two benefits; first, the distance between a hyperplane $(a,b)$ and a point $x$ is $|a^\intercal x -b|$.  
Second, and more importantly, the set of hyperplanes within $r$ of $x$ will necessarily be compact.

\begin{definition}
    The closed ball with radius $r$ centered at $c \in \mathbb{R}^k$ is denoted by $\B(c,r) = \{x \in \mathbb{R}^k: \lVert x-c \rVert_2 \le r \}$.
\end{definition}


\begin{definition}
    For each hyperplane $(a, b) \in {\cal H}$, we denote the left and right half spaces as by $\H^{\twohalfcircle[white]{black}}(a, b):= \{x \in \mathbb{R}^k: a\cdot  x \le b \}$ and $\H^{\twohalfcircle[black]{white}}(a, b):= \{x \in \mathbb{R}^k: a\cdot  x \ge b \}$ respectively.
\end{definition}

A hyperplane in $\mathbb{R}^k$ is considered median hyperplane if it contains at least $|I|/2$ ideal points in each of the two closed half spaces it defines. 
When $|I|$ is even, a median hyperplane doesn't necessarily pass through any ideal points. 
The median hyperplane indicates a position of policy balance, where any move away from it in any direction would leave the majority of voters preferring a position closer to the hyperplane. This balance makes it a set of points of significant strategic interest in political competition.

\begin{definition}
    A hyperplane $\H$ is a median hyperplane if $|\H^{\twohalfcircle[white]{black}} \cap I| \ge |I|/2$ and $|\H^{\twohalfcircle[black]{white}} \cap I| \ge |I|/2$. I.e., $\H$ is a median hyperplane if at least half of the ideal points appear on each side of $\H$.
\end{definition}

The yolk in spatial voting theory is the smallest ball that can be drawn in the policy space such that every median hyperplane intersects the ball. The radius of the yolk reflects the degree of political consensus or conflict: a larger yolk indicates a more dispersed voter, making it harder to find a policy position that a stable majority supports.
\begin{definition}
    A yolk is a smallest ball that intersects with all median hyperplanes.
\end{definition}

\begin{definition}
    In $\mathbb{R}^k$, a hyperplane is a limiting median hyperplane if it passes through at least $k$ ideal points.
\end{definition}

\begin{definition}\label{def: LP yolk}
    A LP yolk is a ball with the smallest radius which intersects all limiting median hyperplanes.
\end{definition}

Example of median hyperplanes and limiting median hyperplanes are shown in Figure \ref{fig:limiting}.

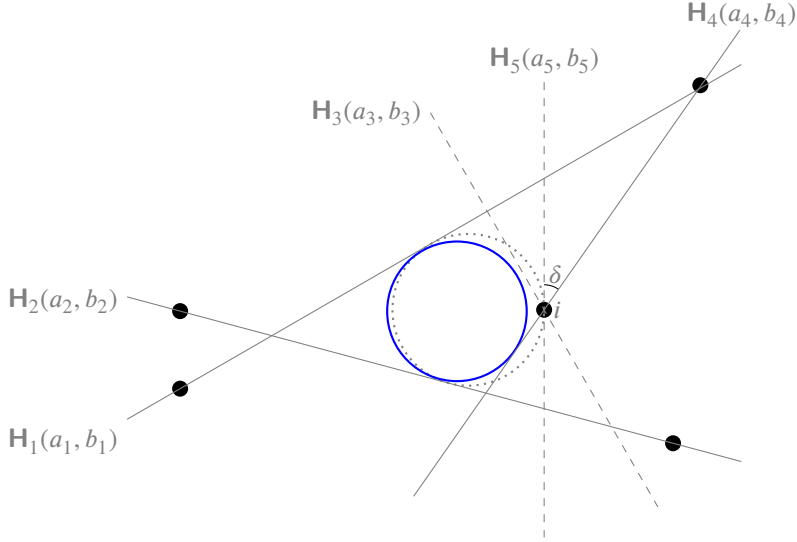
\begin{figure}[!ht]
\centering
\begin{tikzpicture}

\def\a{120}
\def\b{105}
\def\x{3.6}
\def\i{2.7}
\def\e{.5}
\def\y{-4.5}
\def\t{-3.8}
\def\d{30}
\def\k{145}
\def\kk{165}

\pgfmathsetmacro\c{cos(\a)}
\pgfmathsetmacro\s{sin(\a)}
\pgfmathsetmacro\cc{cos(\a/2)/2}
\pgfmathsetmacro\ss{sin(\a/2)/2}
\pgfmathsetmacro\p{-\c/\s*\x+1/\s)}
\pgfmathsetmacro\xx{1/(\c)-\e}
\pgfmathsetmacro\pp{-\c/\s*\y+1/\s}
\pgfmathsetmacro\ppp{-\c/\s*\t+1/\s}

\pgfmathsetmacro\n{cos(\b)}
\pgfmathsetmacro\m{sin(\b)}
\pgfmathsetmacro\qi{\n/\m*\i-1/\m)}
\pgfmathsetmacro\q{\n/\m*\x-1/\m)}
\pgfmathsetmacro\qq{\n/\m*\y-1/\m)}
\pgfmathsetmacro\nn{cos(\b/2)/2}
\pgfmathsetmacro\mm{sin(\b/2)/2}
\pgfmathsetmacro\qqq{\n/\m*\t-1/\m)}

\fill[black] (1,0) circle[radius=3pt];
\fill[black] (\i,\qi) circle[radius=3pt];


\fill[black] (\t,\ppp) circle[radius=3pt];
\fill[black] (\t,\qqq) circle[radius=3pt];
\fill[black] (3.06,2.96) circle[radius=3pt];

\node[ right] at (1,0) {$i$};

\node[below left] at (\y,\pp) {$\H_1(a_1,b_1)$};
\node[above left] at (\y,-\qq) {$\H_2(a_2,b_2)$};
\node[above] at (3.6,3.6) {$\H_4(a_4,b_4)$};
\node[left] at (-0.5,2.598) {$\H_3(a_3,b_3)$};
\node[above] at (1,3) {$\H_5(a_5,b_5)$};

\draw (\x,\p)--(\y,\pp);
\draw (\x,\q)--(\y,\qq);
\draw[dashed] (1,3)--(1,-3);
\draw[dotted,thick] (0,0) circle[radius=1];
\draw[blue,thick] (-0.15,-0.02) circle[radius=0.92];




\coordinate (o) at (1,0);
\coordinate (a) at (3.6,3.6);
\coordinate (b) at (1,3);
\pic["$\delta$", draw=black, -, angle eccentricity=1.5, angle radius=0.33cm]
    {angle=a--o--b};

\draw[rotate around={\d:(1,0)}, dashed] (1,3)--(1,-3);
\draw[rotate around={\k:(1,0)}] (1,3)--(1,-4.5);

\end{tikzpicture}\caption{Examples of limiting median hyperplanes and non-limiting median hyperplanes in $2$-dimension space. The black points are voters, the dotted black circle is yolk and the blue circle is LP yolk. The hyperplane $\H_1(a_1,b_1), \H_2(a_2,b_2)$ and $\H_4(a_4,b_4)$ are limiting median hyperplanes, and $\H_3(a_3,b_3), \H_5(a_5,b_5)$ are non-limiting median hyperplanes. $\H_ 3(a_3,b_3)$ rotates $\delta$ with $i$ as the center, that is, $\H_ 4(a_4,b_4)$.
}\label{fig:limiting}
\end{figure}




\section{Relationship Between Yolks and LP Yolks: The Worst Case}


In this section, we provide families of examples demonstrating that the LP yolk is a poor approximation of the yolk. 
We prove that when the dimension is $k\ge3$, or when both $k=2$ and the number of voters is even, the LP yolk can be arbitrarily small relative to the yolk (Theorem \ref{thm:EvenBad}).
However, we show that for $k=2$ and $|I|$ is odd, that there exists a family of examples where the LP yolk radius is at least 1/2 of the yolk radius (Theorem \ref{thm:OddR2OK}). 
Later, in Section \ref{sec:Odd}, we provide a positive result on the approximately of the yolk by showing that this bound of 1/2 is tight. 
However, in Theorem \ref{thm:OddR2Far} in this section, we show that despite the fact the yolk radius is similar in size, that it can be arbitrarily far from the yolk center, when $k=2$ and the number of voters is odd (Theorem \ref{thm:OddR2Far}).

\begin{theorem}\label{thm:EvenBad}
    Let $I\subset \mathbb{R}^k$. 
    The LP yolk radius can be arbitrarily small relative to the yolk radius when (1) $k\geq 3$ or  (2) $k=2$ and $|I|$ is even. 
\end{theorem}

\begin{proof}
    Consider any set of ideal points $I'\subset \mathbb{R}^{k-1}$ with yolk radius $r>0$. 
    Such a set of ideal points exists since (1) the core is almost always empty when $k-1\geq 2$ and (2) there is almost never a unique median when there are an even number of voters
    (see e.g., \cite{martin2019dominance}).

    Let $I=\{(\pi,0): \pi\in I'\}\subset \mathbb{R}^k$.
    It is straightforward to verify that if $\{x\in \mathbb{R}^{k-1}: \sum_{i=1}^{k-1} a_i\cdot x_i = b\}$ is a median hyperplane for $I'$, then $\{x\in \mathbb{R}^{k}: \sum_{i=1}^{k-1} a_i\cdot x_i +0\cdot x_k= b\}$ is a median hyperplane for $I$.
    The yolk is at least large enough to cover all of these median hyperplanes. 
    Following from McKelvey's linear programming description of the yolk, the yolk radius is at least $r$. 
    
    Since all points kth coordinate is zero, they are contained in the hyperplane $\{x\in \mathbb{R}^k: x_k=0\}$, and the only possible limiting median hyperplane is $\{x\in \mathbb{R}^k: x_k=0\}$ implying that the LP yolk radius is 0. 
    Therefore the LP yolk radius can be arbitrarily small relative to the yolk radius. 
\end{proof}

While our construction in Theorem \ref{thm:EvenBad} is degenerate, there exists non-degenerate instances where the LP yolk radius can be arbitrarily small.
Instead of appending each ideal point with $0$, append each ideal point with $\epsilon$-noise.  
As $\epsilon\to 0$, we recover the instance in the proof of Theorem \ref{thm:EvenBad}. 
Thus, when $k\geq 3$ or when $k=2$ and $|I|$ is even, there is measurable set of instances where the LP yolk is significantly smaller than the corresponding yolk.  
For completeness, we provide a non-degenerate example for $k=2$ in Proposition \ref{prop:NonDegen}.

\begin{proposition}[First observed in \cite{stone1992limiting}]\label{prop:NonDegen}
    Let $I\subset \mathbb{R}^2$ where $|I|$ is even. 
    The LP yolk radius can be arbitrarily small relative to the yolk radius even when no 3 points are colinear.
\end{proposition}

\begin{proof}
We examine the set of six ideal points given by $I=\{(\pm 2, \pm \epsilon),(\pm 1, 0)\}$ where $\epsilon>0$.  
We show the yolk (depicted in Figure \ref{fig:yolk2}) has radius at least  $1$ while the LP yolk (Figure \ref{fig:lpyolk2}) has radius $\frac{\epsilon}{1+\epsilon^2}\to 0$ as $\epsilon \to 0$ thereby completing the proposition.
We remark that our construction with $\epsilon=1/2$ matches the construction given in \cite{stone1992limiting}, which establishes that the yolk and LP yolk are distinct concepts. 
Further, \cite{stone1992limiting} similarly acknowledges that the LP yolk becomes arbitrarily small relative to the yolk as the points are forced to be colinear.

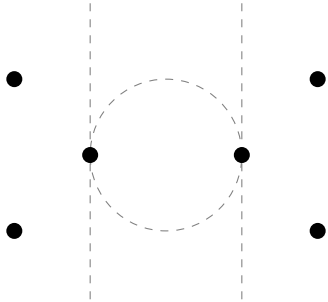
\begin{figure}[!ht]
\centering
\begin{tikzpicture}
\draw[dashed] (0,0) circle[radius=1];
\draw[dashed] (1,2)--(1,-2);
\draw[dashed] (-1,2)--(-1,-2);

\fill[black] (1,0) circle[radius=3pt];
\fill[black] (-1,0) circle[radius=3pt];
\fill[black] (2,1) circle[radius=3pt];
\fill[black] (2,-1) circle[radius=3pt];
\fill[black] (-2,1) circle[radius=3pt];
\fill[black] (-2,-1) circle[radius=3pt];

%
%
\end{tikzpicture}\caption{A set of non-degenerate ideal points yielding an arbitrarily small LP Yolk.}\label{fig:yolk2}
\end{figure}

Both $\{x: x_1\geq 1\}$ and $\{x: x_1\leq -1\}$ are median hyperplanes and any ball containing both median hyperplanes must have radius at least one. 
Therefore the yolk radius is at least one.
We remark that the yolk radius is exactly one, but for the proposition, it suffices to only show the radius is at least one. 

We now show that the LP yolk radius is $\frac{\epsilon}{\sqrt{1+\epsilon^2}}$.  
By symmetry, $\B((0,0),r)$ is an LP yolk where $r$ is the maximum distance to a limiting median hyperplane.  
There are 11 limiting median hyperplanes which are given in Figure \ref{fig:lpyolk2}.

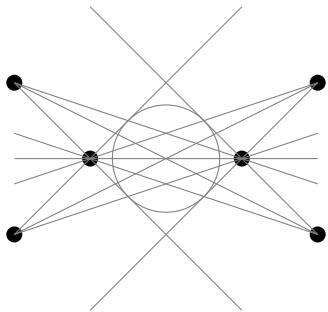
\begin{figure}[!ht]
\centering
\begin{tikzpicture}

\fill[black] (1,0) circle[radius=3pt];
\fill[black] (-1,0) circle[radius=3pt];
\fill[black] (2,1) circle[radius=3pt];
\fill[black] (2,-1) circle[radius=3pt];
\fill[black] (-2,1) circle[radius=3pt];
\fill[black] (-2,-1) circle[radius=3pt];

\draw (2,1)--(-2,-.3333);
\draw (2,-1)--(-2,.3333);
\draw (-2,1)--(2,-.3333);
\draw (-2,-1)--(2,.3333);

\draw (2,1)--(-2,-1);
\draw (2,-1)--(-2,1);
\draw (2,0)--(-2,0);

\draw (2,1)--(-1,-2);
\draw (-2,1)--(1,-2);
\draw (2,-1)--(-1,2);
\draw (-2,-1)--(1,2);

\draw (0,0) circle[radius=.707];
\end{tikzpicture}\caption{Set of limiting median hyperplanes and the resulting LP Yolk.}\label{fig:lpyolk2}
\end{figure}

By symmetry, it suffices to consider only the limiting median hyperplane $\{x: x_2=0\}$ and the limiting median hyperplanes going through $(2,\epsilon)$ --- $\{x: \epsilon\cdot  x_1-x_2=\epsilon\}$, $\{x: - x_1+3x_2=3\epsilon\}$ and $\{x: \epsilon \cdot x_1-2x_2=0\}$. The distance from $(0,0)$ to each of these median hyperplanes respectively is $0, \frac{\epsilon}{\sqrt{1+\epsilon^2}}$, $\frac{\epsilon}{\sqrt{9+\epsilon^2}}$, and $0$. 
Therefore the LP yolk has radius $\frac{\epsilon}{\sqrt{1+\epsilon^2}}\to 0$ as $\epsilon\to 0$. 
Since this hold for all $\epsilon$, the LP yolk can be arbitrarily small relative to the size of the yolk. 
\end{proof}

Despite these negative results in the general setting, the spatial model voting is often used in two-dimensional settings to address small, typically odd-sized committees --- odd-sized committees avoid ties. 
Thus, it remains important to also understand when $I\subset \mathbb{R}^2$ and $|I|$ is odd. 
We begin by showing that there exists a family of examples where the LP yolk radius can be half the size of the yolk (Theorem \ref{thm:OddR2OK}). 
Later in Theorem \ref{thm:MainHalf}, we show this bound is tight and the LP yolk radius is a reasonable approximation of the yolk radius in this setting. 
However, we also show that the LP yolk center can be arbitrarily far away from the yolk center even when their radii are similar (Theorem \ref{thm:OddR2Far}). 

\begin{theorem}\label{thm:OddR2OK}
    When $|I|$ is odd in $\mathbb{R}^2$, there exists a set of ideal points where the ratio between the LP yolk radius is at most $1/2+\epsilon$ the size of the yolk radius for all $\epsilon >0$. 
\end{theorem} 

\begin{proof}
Consider the five ideal points $x^0=(1,0), x^{\pm 1}=\left(w,\pm \frac{1-\cos(\alpha)w}{\sin(\alpha)}\right), x^{\pm 2}=\left(\frac{1}{\cos(\alpha)}-\epsilon,\pm\frac{\epsilon \cos(\alpha)}{\sin(\alpha)}\right)$ for $\alpha\in \left(\frac{\pi}{2},\pi\right)$, $w>1$ and $\epsilon$ small as shown in Figure \ref{fig:yolkodd}.

\def\a{97}
\def\x{5}
\def\e{1}

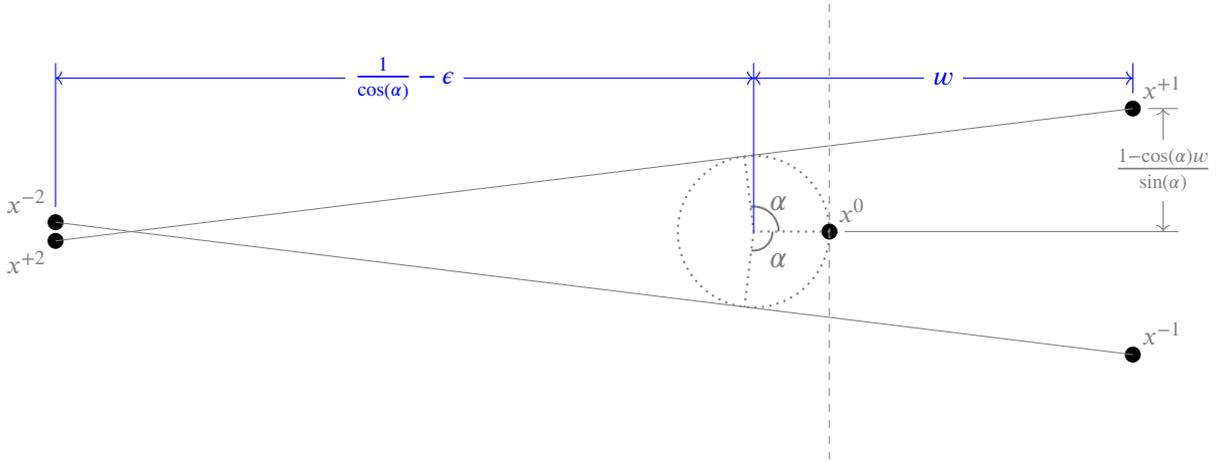
\begin{figure}[!ht]
\centering
\begin{tikzpicture}

\pgfmathsetmacro\c{cos(\a)}
\pgfmathsetmacro\s{sin(\a)}
\pgfmathsetmacro\cc{cos(\a/2)/2}
\pgfmathsetmacro\ss{sin(\a/2)/2}
\pgfmathsetmacro\p{-\c/\s*\x+1/\s)}
\pgfmathsetmacro\xx{1/(\c)-\e}
\pgfmathsetmacro\pp{-\c/\s*\xx+1/\s}

\fill[black] (1,0) circle[radius=3pt];
\fill[black] (\x,\p) circle[radius=3pt];
\fill[black] (\x,-\p) circle[radius=3pt];

\node[above right] at (1,0) {$x^0$};
\node[above right] at (\x,\p) {$x^{+1}$};
\node[above right] at (\x,-\p) {$x^{-1}$};

\fill[black] (\xx,\pp) circle[radius=3pt];
\fill[black] (\xx,-\pp) circle[radius=3pt];

\node[below left] at (\xx,\pp) {$x^{+2}$};
\node[above left] at (\xx,-\pp) {$x^{-2}$};

\draw (\x,\p)--(\xx,\pp);
\draw (\x,-\p)--(\xx,-\pp);

\draw[dotted,thick] (0,0) circle[radius=1];

\node at (\cc,\ss) {\large$\alpha$};
\draw[thick] (.33,0) arc (0:\a:.33);

\node at (\cc,-\ss) {\large$\alpha$};
\draw[thick] (.25,0) arc (0:-\a:.25);

\draw[dotted,thick] (1,0)--(0,0);
\draw[dotted,thick] (\c,\s)--(0,0);
\draw[dotted,thick] (0,0)--(\c,-\s);


\draw[blue] (0,0)--($(0,\p+.6)$);
\draw[blue] ($(\x,\p+.2)$)--($(\x,\p+.6)$);
\draw[<->,blue] ($(0,\p+.4)$)--($(\x,\p+.4)$) node[midway, fill=white] {$w$} ;

\draw[<->,blue] ($(0,\p+.4)$)--($(\xx,\p+.4)$) node[midway, fill=white] {$\frac{1}{\cos(\alpha)}-\epsilon$} ;
\draw[blue] ($(\xx,-\pp+.2)$)--($(\xx,\p+.6)$);

\draw ($(\x+.2,\p)$)--($(\x+.6,\p)$);
\draw ($(1.2,0)$)--($(\x+.6,0)$);
\draw[<->] ($(\x+.4,\p)$)--($(\x+.4,0)$) node[midway, fill=white] {$\frac{1-\cos(\alpha)w}{\sin(\alpha)}$} ;

\draw[dashed] (1,3)--(1,-3);
\end{tikzpicture}\caption{Set of ideal points and the yolk yielding a ratio of $\frac{1}{2}+\epsilon$.}\label{fig:yolkodd}
\end{figure}

The ideal points are selected such that $x^0$ is the right most point on the unit circle $B((0,0),1)$, the lines $\overline{x^{+1}x^{+2}}$ and $\overline{x^{-1}x^{-2}}$ are tangent to the unit circle at $(\cos(\alpha),\sin(\alpha))$ and $(\cos(\alpha),-\sin(\alpha))$ respectively. 
By construction, $\overline{x^{+1}x^{-2}}$ and $\overline{x^{-1}x^{+2}}$ are limiting median lines while the dotted line $\{x: x_2=1\}$ is a non-limiting median line and therefore the yolk radius is at least one.  
It also straightforward to verify that the unit circle $\B((0,0),1)$ is the unique yolk and has radius one.

There are $6$ limiting median lines, $\overline{x^{+1}x^{+2}}, \overline{x^{-1}x^{-2}}, \overline{x^{+1}x^{0}}, \overline{x^{-1}x^{0}}, \overline{x^{+2}x^{0}},$ and $\overline{x^{-2}x^{0}}$ as depicted in Figure \ref{fig:LB.Limiting}.

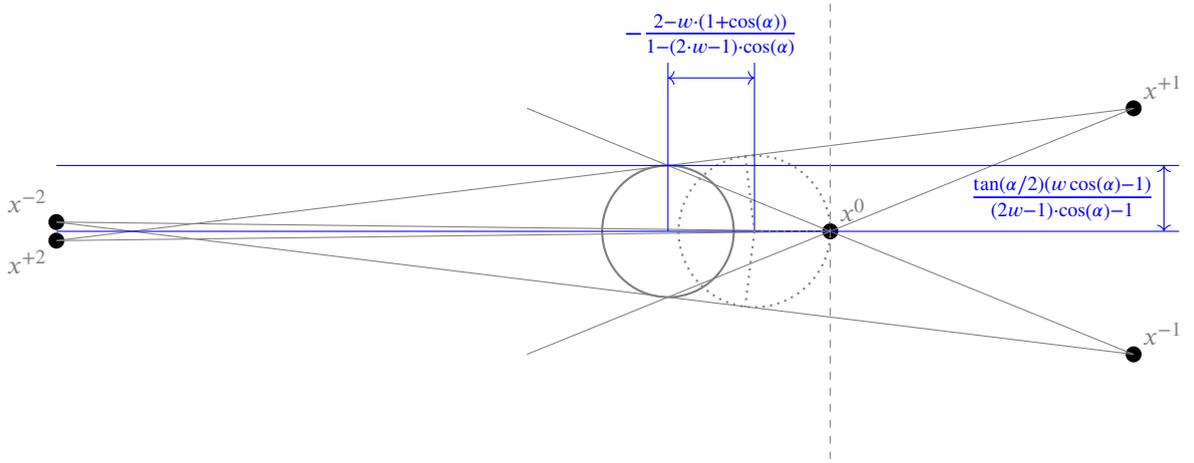
\begin{figure}[!ht]
\centering
\begin{tikzpicture}

\pgfmathsetmacro\c{cos(\a)}
\pgfmathsetmacro\s{sin(\a)}
\pgfmathsetmacro\cc{cos(\a/2)/2}
\pgfmathsetmacro\ss{sin(\a/2)/2}
\pgfmathsetmacro\p{-\c/\s*\x+1/\s)}
\pgfmathsetmacro\xx{1/(\c)-\e}
\pgfmathsetmacro\pp{-\c/\s*\xx+1/\s}

\pgfmathsetmacro\cc{(2-\x*(1+cos(\a)))/(1-(2*\x-1)*cos(\a))}
\pgfmathsetmacro\rr{(tan(\a/2)*(\x*cos(\a)-1))/(2*\x*cos(\a)-cos(\a)-1)}

\fill[black] (1,0) circle[radius=3pt];
\fill[black] (\x,\p) circle[radius=3pt];
\fill[black] (\x,-\p) circle[radius=3pt];

\node[above right] at (1,0) {$x^0$};
\node[above right] at (\x,\p) {$x^{+1}$};
\node[above right] at (\x,-\p) {$x^{-1}$};

\fill[black] (\xx,\pp) circle[radius=3pt];
\fill[black] (\xx,-\pp) circle[radius=3pt];

\node[below left] at (\xx,\pp) {$x^{+2}$};
\node[above left] at (\xx,-\pp) {$x^{-2}$};

\draw (\x,\p)--(\xx,\pp);
\draw (\x,-\p)--(\xx,-\pp);

\draw (\x,\p)-- (2-\x,-\p);
\draw (\x,-\p)-- (2-\x,\p);

\draw (\xx,\pp)--(1,0)--(\xx,-\pp);

\draw[dotted,thick] (0,0) circle[radius=1];

\draw[thick] (\cc,0) circle[radius=\rr];

\draw[blue] (0,0)--($(0,\p+.6)$);
\draw[<->, blue] ($(0,\p+.4)$)--($(\cc,\p+.4)$) node[above, midway,yshift=1ex] {$-\frac{2-w\cdot\left(1+\cos\left(\alpha\right)\right)}{1-\left(2\cdot w-1\right)\cdot\cos\left(\alpha\right)}$} ;
\draw[blue] ($(\cc,0)$)--($(\cc,\p+.6)$);

\draw[blue] ($(\xx,\rr)$)--($(\x+.6,\rr)$);
\draw[blue] ($(\xx,0)$)--($(\x+.6,0)$);
\draw[<->,blue] ($(\x+.4,\rr)$)--($(\x+.4,0)$) node[left, midway] {$\frac{\tan\left(\alpha/2\right)\left(w\cos\left(\alpha\right)-1\right)}{(2w-1)\cdot\cos\left(\alpha\right)-1}$} ;

\draw[dotted,thick] (1,0)--(0,0);
\draw[dotted,thick] (\c,\s)--(0,0);
\draw[dotted,thick] (0,0)--(\c,-\s);

\draw[dashed] (1,3)--(1,-3);
\end{tikzpicture}\caption{Set of limiting median lines and a ball that intersects them all. 
The measurements indicate the distance between the two balls and the height of the new ball that intersects all limiting median hyperplanes. }\label{fig:LB.Limiting}
\end{figure}

The equations of the lines $\overline{x^{+1}x^{+2}}$ and $\overline{x^{-1}x^{0}}$ are given by
\begin{align*}
    \sin(\alpha)\cdot x_2 &= - \cos(\alpha)\cdot x_1 +1 \tag{$\overline{x^{+1}x^{+2}}$}\\
    \left(w-1\right)\cdot\sin\left(\alpha\right)\cdot x_2&=-\left(1-w\cdot\cos\left(\alpha\right)\right)\left(x_1-1\right) \tag{$\overline{x^{-1}x^{0}}$}
\end{align*}
which intersect at 
\begin{align*}
    (c,r):=\left( \frac{2-w\cdot\left(1+\cos\left(\alpha\right)\right)}{1-\left(2\cdot w-1\right)\cdot\cos\left(\alpha\right)},  \frac{\tan\left(\alpha/2\right)\left(w\cos\left(\alpha\right)-1\right)}{(2w-1)\cdot\cos\left(\alpha\right)-1} \right).  \tag{Intersection of $\overline{x^{+1}x^{+2}}$ and $\overline{x^{-1}x^{0}}$}
\end{align*}
Therefore, the ball centered at $(c, 0)$ with radius $r$ intersects both of these median lines.  
Symmetrically, the ball also intersects $\overline{x^{-1}x^{-2}}$ and $\overline{x^{+1}x^{0}}$. 
Finally, for sufficiently small values of $\epsilon$, the ball also trivially intersects $\overline{x^{\pm 2}x^0}$. 
Thus, $\B(c,r)$ is a ball intersecting all limiting median hyperplanes and the LP yolk has radius at most $r$. 
Since the yolk radius is one, the LP yolk radius is at most $r$ relative to the yolk radius.  
In particular for $w=-\frac{k}{\cos(\alpha)}$ for $k\in \mathbb{R}_{>0}$ where $\frac{k+1}{2k+1}<1/2+\epsilon$,
\begin{align*}
    r = \frac{\tan\left(\alpha/2\right)\left(w\cos\left(\alpha\right)-1\right)}{(2w-1)\cdot\cos\left(\alpha\right)-1}=\frac{(k+1)\cdot \tan(\alpha/2)}{2k+1+\cos(\alpha)}  \to \frac{k+1}{2k+1} \ as \ \alpha\to\pi/2.
\end{align*}
yielding the statement of the theorem.  
\end{proof}

While the bound of 1/2, which we later show is tight, seems promising, the center of LP yolk can be arbitrarily far from the center of yolk.
This suggests that the LP yolk, even in this simple case, can be a poor approximation of the yolk.

\begin{theorem}\label{thm:OddR2Far}
    Let $I\subset \mathbb{R}^2$ where $|I|$ is odd. 
    Then the LP yolk can be arbitrarily far from the yolk (relative to the yolk radius). 
\end{theorem}

\begin{proof}

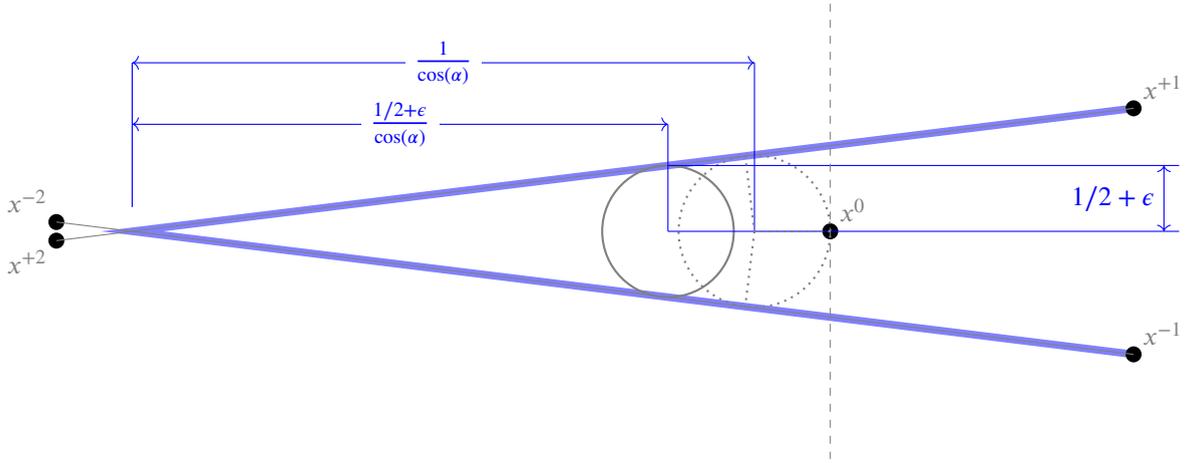
\begin{figure}[!ht]
\centering
\begin{tikzpicture}

\def\a{97}
\def\x{5}
\def\e{1}

\pgfmathsetmacro\c{cos(\a)}
\pgfmathsetmacro\s{sin(\a)}
\pgfmathsetmacro\cc{cos(\a/2)/2}
\pgfmathsetmacro\ss{sin(\a/2)/2}
\pgfmathsetmacro\p{-\c/\s*\x+1/\s)}
\pgfmathsetmacro\xx{1/(\c)-\e}
\pgfmathsetmacro\xxx{1/(\c)}
\pgfmathsetmacro\pp{-\c/\s*\xx+1/\s}

\pgfmathsetmacro\cc{(2-\x*(1+cos(\a)))/(1-(2*\x-1)*cos(\a))}
\pgfmathsetmacro\rr{(tan(\a/2)*(\x*cos(\a)-1))/(2*\x*cos(\a)-cos(\a)-1)}

\draw[line width=1mm, blue!50] (\x,\p)--(\xxx,0)--(\x,-\p);

\fill[black] (1,0) circle[radius=3pt];
\fill[black] (\x,\p) circle[radius=3pt];
\fill[black] (\x,-\p) circle[radius=3pt];

\node[above right] at (1,0) {$x^0$};
\node[above right] at (\x,\p) {$x^{+1}$};
\node[above right] at (\x,-\p) {$x^{-1}$};

\fill[black] (\xx,\pp) circle[radius=3pt];
\fill[black] (\xx,-\pp) circle[radius=3pt];

\node[below left] at (\xx,\pp) {$x^{+2}$};
\node[above left] at (\xx,-\pp) {$x^{-2}$};

\draw (\x,\p)--(\xx,\pp);
\draw (\x,-\p)--(\xx,-\pp);

\draw[dotted,thick] (0,0) circle[radius=1];

\draw[thick] (\cc,0) circle[radius=\rr];

\draw[blue] (\cc,0)--($(\cc,\p/2+.6)$);
\draw[<->,blue] ($(\cc,\p/2+.6)$)--($(\xxx,\p/2+.6)$) node[midway, fill=white] {$\frac{1/2 + \epsilon}{\cos(\alpha)}$} ;
\draw[blue] ($(\xxx,-\pp+.2)$)--($(\xxx,\p+.6)$);

\draw[blue] (0,0)--($(0,\p+.6)$);
\draw[<->,blue] ($(0,\p+.6)$)--($(\xxx,\p+.6)$) node[midway, fill=white] {$\frac{1}{\cos(\alpha)}$} ;
\draw[blue] ($(\xxx,-\pp+.2)$)--($(\xxx,\p+.6)$);


\draw[blue] ($(\cc,\rr)$)--($(\x+.6,\rr)$);
\draw[blue] ($(\cc,0)$)--($(\x+.6,0)$);
\draw[<->,blue] ($(\x+.4,\rr)$)--($(\x+.4,0)$) node[left, midway,blue] {$1/2+\epsilon$} ;

\draw[dotted,thick] (1,0)--(0,0);
\draw[dotted,thick] (\c,\s)--(0,0);
\draw[dotted,thick] (0,0)--(\c,-\s);

\draw[dashed] (1,3)--(1,-3);
\end{tikzpicture}\caption{By proportionality of similar triangles, the center of the LP yolk must be at least $\frac{1/2-\epsilon}{\cos(\alpha)}$ from the center of the yolk. }\label{fig:LB.Limiting2}
\end{figure}
The proof follows using the same family of instances in the proof of Theorem \ref{thm:OddR2OK}. 
Consider the cone given by $\overline{x^{\pm 1}, x^{\pm 2}}$ as depicted in Figure \ref{fig:LB.Limiting2} and let $(1/\cos(\alpha),0)$ denote the point of the cone.
Selecting $k$ sufficiently small and $\alpha$ sufficiently close to $\alpha/2$, the LP yolk is a ball tangent to the cone with radius $r\leq 1/2+\epsilon$ while the yolk is a ball tangent to the cone with radius $r'= 1$. 
Let $d$ and $d'$ be the distance from the point of the cone to the LP yolk center and the yolk center respectively; therefore the center of the LP yolk is $d'-d$ away from the center of the yolk. 
By proportionality of similar triangles, $d\leq (1/2+\epsilon)\cdot d'$ and, by construction, $d'= 1/\cos(\alpha)$.  
Finally, as $\alpha\to \pi/2$, $d'-d\geq \frac{1/2-\epsilon}{\cos(\alpha)}\to \infty$ thereby completing the proof. 
\end{proof}

We remark that all of constructions feature points in general position --- no set of three points are colinear --- and therefore bad instances are likely to occur with positive probability when ideal points are generated from a continuous distribution. 
In the remainder of the paper, we focus on showing that the bound of 1/2 is tight. 

\section{Properties of Yolks and LP Yolks}

In this section, we establish several properties of yolks to assist in showing that the LP yolk radius is at least 1/2 the yolk radius when $|I|\in \mathbb{R}^2$ and $|I|$ is odd.
We remark that the properties within may be useful for establishing other results regarding the yolk. 
Most notably, we show that at most $k+1$ median hyperplanes in $\mathbb{R}^k$ are needed to determine the yolk --- a drastic improvement on the bound of $O(|I|^k)$ given by \cite{Tovey92}.

\begin{definition}
   Given the ball $\B(c,r) = \{x \in \mathbb{R}^k: \lVert x-c \rVert_2 \le r \}$ with radius $r$ centered at $c$, for $\alpha$ where $\lVert \alpha \rVert_2 = 1$, the $\alpha$-hemisphere of $\B(c, r)$ is $\He(\alpha,c,r) = \{ x \in ext(\B(c, r)): a^t x \ge a^t c \}$ where $ext(\B(c, r)) = \{x \in \mathbb{R}^k: \lVert x-c \rVert_2 = r\}$ is the exterior/surface of $\B(c, r)$.
\end{definition}


As indicated in Figure \ref{fig:shift}, the hyperplane $\H_1(a_1,b_1)$ is a median hyperplane tangent to ball $\B$ at $(a_1,b_1)$, and $\He(\alpha,c,r)$ is the $\alpha$-hemisphere of $\B(c, r)$.

Bailey \cite{bailey2023dimension} shows that for any unit vector $\alpha$ and for any yolk $\B(c,r)$, there is at least one median hyperplane tangent to the $\alpha$-hemisphere $\He(\alpha,c,r)$. 
We establish that this provides a complete characterization of the yolk by showing the second direction:  $\B(c, r)$ is a yolk if and only if every $\alpha$-hemisphere $\He(\alpha,c,r)$ has at least one tangent median hyperplane.

\begin{lemma}\label{lem:characterization}
    Let $\cal H$ be a compact set of hyperplanes. Then $\B(c,r)$ is a smallest ball intersecting every hyperplane  in ${\cal H}$ if and only if every $\alpha$-hemisphere of $\B(c,r)$ has at least one tangent hyperplane in ${\cal H}$. 
\end{lemma}

Notably, \cite{bailey2023dimension} shows that the set of median hyperplanes is compact, and therefore Lemma \ref{lem:characterization} is a characterization for a yolk.  
However, for our later set of results, we consider an arbitrary compact set of hyperplanes, not necessarily the set of median hyperplanes. 

\begin{proof}[Proof of Lemma \ref{lem:characterization}]
    The first direction is given in \cite{bailey2023dimension}.  
    For completeness, we provide a sketch of the proof in the caption of Figure \ref{fig:shift}. 

    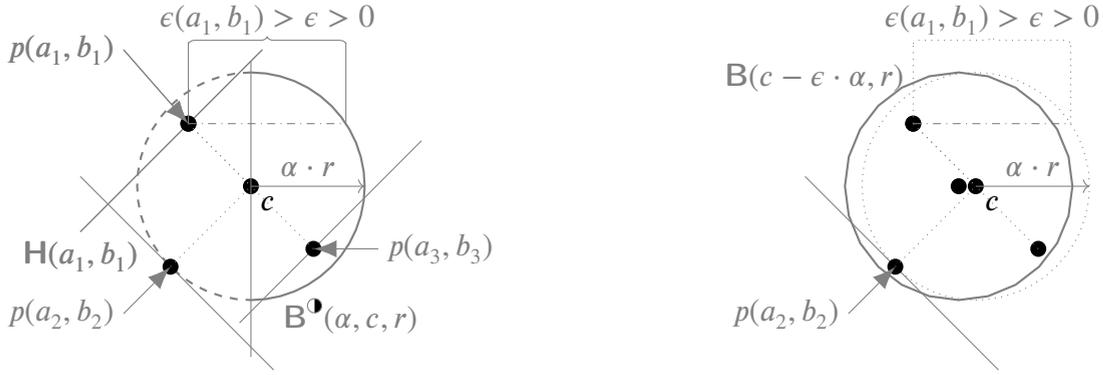
\begin{figure}[!ht]
    \begin{center}
    \begin{tikzpicture}[scale=1.5]
        \draw [thick,domain=-90:90,white] plot ({cos(\x)+1.385}, {sin(\x)});  
        \draw [thick,domain=90:270,white] plot ({cos(\x)+1.385}, {sin(\x)});
        \draw [ thick,domain=-90:90] plot ({cos(\x)}, {sin(\x)});  
        \draw [dashed,thick,domain=90:270] plot ({cos(\x)}, {sin(\x)});
        \fill[black] (0,0) circle[radius=2pt] node[below right] {\large $c$};
        \draw (0,1.1)--(0,-1.5);
        \node[below right] at (.2,-.9) {\large $\He(\alpha,c,r)$};
         
        \draw [domain=-.1:1.5] plot({-\x}, {-\x+1.1}) node[below] {\large $\H(a_1,b_1)$};
        \fill[black] (-.55,.55) circle[radius=2pt];
        \draw[dotted] (-.55,.55)--(0,0);
        \draw[latexnew-,arrowhead=.12in] (-.55,.55)--(-1.12,1.2) node[left] {\large $p(a_1,b_1)$};
         
        \draw [domain=-.1:1.5] plot({\x}, {\x-1.1});
        \fill[black] (.55,-.55) circle[radius=2pt];
        \draw[dotted] (.55,-.55)--(0,0);
        \draw[latexnew-,arrowhead=.12in] (.55,-.55)--(1.12,-.55) node[right] {\large $p(a_3,b_3)$};

        \draw [domain=-.1:1.5] plot({-\x}, {-\x+1.1}) node[below] {\large $\H(a_1,b_1)$};
        \fill[black] (-.55,.55) circle[radius=2pt];
        \draw[dotted] (-.55,.55)--(0,0);
        \draw[latexnew-,arrowhead=.12in] (-.55,.55)--(-1.12,1.2) node[left] {\large $p(a_1,b_1)$};

        \draw [domain=-1.5:.2] plot({\x}, {-\x-1.414});
        \fill[black] (-.707,-.707) circle[radius=2pt];
        \draw[dotted] (-.707,-.707)--(0,0);
        \draw[latexnew-,arrowhead=.12in] (-.707,-.707)--(-1.12,-1.12) node[left] {\large $p(a_2,b_2)$};

        \draw[dashdotted] (-.55,.55)--(.835,.55);
        \draw[decoration={brace,raise=30pt},decorate]
      (-.55,.55) -- node[above=30pt] {\large $\epsilon(a_1,b_1)> \epsilon>0$} (.835,0.55);
        \draw(-.55,.55)--(-.55,1.265);
        \draw(.835,.55)--(.835,1.265);
        \draw[->] (0,0)--(1,0) node[midway, above] {\large $\alpha\cdot r$};
    \end{tikzpicture}
    \hfill
    \begin{tikzpicture}[scale=1.5]
        \draw [thick,domain=-90:90,white] plot ({cos(\x)+1.385}, {sin(\x)});  
        \draw [thick,domain=90:270,white] plot ({cos(\x)+1.385}, {sin(\x)});
        \draw [dotted,domain=-90:90] plot ({cos(\x)}, {sin(\x)});  
        \draw [dotted,domain=90:270] plot ({cos(\x)}, {sin(\x)});
        
        \draw [thick,domain=0:360] plot ({cos(\x)-0.15}, {sin(\x)});
        \fill[black] (-0.15,0) circle[radius=2pt];
        \node[above left] at (-.55,.75) {\large $\B(c-\epsilon\cdot \alpha,r)$};

        \fill[black] (0,0) circle[radius=2pt] node[below right] {\large $c$};
         
        \fill[black] (-.55,.55) circle[radius=2pt];
        \draw[dotted] (-.55,.55)--(0,0);
         
        \fill[black] (.55,-.55) circle[radius=2pt];
        \draw[dotted] (.55,-.55)--(0,0);

        \fill[black] (-.55,.55) circle[radius=2pt];
        \draw[dotted] (-.55,.55)--(0,0);

        \draw [domain=-1.5:.2] plot({\x}, {-\x-1.414});
        \fill[black] (-.707,-.707) circle[radius=2pt];
        \draw[dotted] (-.707,-.707)--(0,0);
        \draw[latexnew-,arrowhead=.12in] (-.707,-.707)--(-1.12,-1.12) node[left] {\large $p(a_2,b_2)$};

        \draw[dashdotted] (-.55,.55)--(.835,.55);
        \draw[dotted, decoration={brace,raise=30pt},decorate]
      (-.55,.55) -- node[above=30pt] {\large $\epsilon(a_1,b_1)> \epsilon>0$} (.835,0.55);
        \draw[dotted](-.55,.55)--(-.55,1.265);
        \draw[dotted](.835,.55)--(.835,1.265);
        \draw[->] (0,0)--(1,0) node[midway, above] {\large $\alpha\cdot r$};
    \end{tikzpicture}
    \caption{
    Suppose there is some hemisphere $\He(\alpha,c,r)$ with no tangent hyperplane. 
    Let $p(a,b)$ be the closest point on the hyperplane $(a,b)$ to the center $c$. 
    Since $p(a,b)$ is not on this hemisphere, there is some point strictly to the right on the hemisphere $\He(\alpha,c,r)$. 
    Equivalently, there exists an $\epsilon>0$ such that the ball $\B(\alpha, c-\epsilon\cdot\alpha, r)$ contains $p(a,b)$ in the interior. 
    This holds for every $(a,b)$ in the compact set ${\cal H}$, there exists some $r'<r$ such that every $p(a,b)$ is contained in $\B(c-\epsilon\cdot\alpha, r')$, i.e., if the $\alpha$-hemisphere is uncovered, then we can find a smaller ball intersecting every hyperplane in ${\cal H}$.
    The full details are given in \cite{bailey2023dimension}.
    }\label{fig:shift}
    \end{center}
\end{figure}

    For the second direction, we assume $\B(c,r)$ intersects every hyperplane in $\cal H$ and for every $\alpha$-hemisphere that there is a tangent hyperplane in $\cal H$. 
    We show that $\B(c,r)$ is a smallest ball intersecting hyperplane in $\cal H$. 
    Consider any ball $\B(c+{\alpha},r')$ that intersects with every hyperplane in $\cal H$. 
    It suffices to show that $r'\geq r$.

    First, suppose ${\alpha}=\vec{0}$. 
    Since there is a hyperplane $(a,b)$ tangent to $\B(c,r)$, the hyperplane $(a,b)$ is exactly $r$ away from $c$ implying that $r'\geq r$ since $B(c,r')$ intersects with the hyperplane given by $(a,b)$.
    
    Next, suppose ${\alpha}\neq\vec{0}$ and $||{\alpha}||_2>0$, and consider the $-\frac{\alpha}{||\alpha||_2}$-hemisphere of $\B(c,r)$, $\He(-\frac{\alpha}{||\alpha||_2}, c,r)$. 
    By selection of $\B(c,r)$, there exist a $(a,b)\in {\cal H}$ such that $\H(a,b)$ is tangent to $\He(-\frac{\alpha}{||\alpha||_2},c,r)$.
    Since $a$ is a unit vector and since $\H(a,b)$ is exactly $r$ away from $c$, $a^\intercal c \in \{b+r,b-r\}$.     
    Without loss of generality, we assume $a^\intercal c=b-r$ --- in the second case, the proof follows identically by considering the equivalent hyperplane $(-a,-b)$. 
    Let $p(a,b)=H(a,b)\cap \He(-\frac{\alpha}{||\alpha||_2},c,r)$ be the closest point in $H(a,b)$ to $c$. 
    Since $a^\intercal c=b-r$ and since $\H(a,b)$ is tangent to $B(c,r)$, $p(a,b)=c+r\cdot a$.
    Further, since $p(a,b)\in \He(-\frac{\alpha}{||\alpha||_2},c,r)$, 
    \begin{align*}
        -\frac{\alpha}{||\alpha||_2}^\intercal p(a,b) &\geq -\frac{\alpha}{||\alpha||_2}^\intercal c\\
        \Rightarrow-\alpha^\intercal (c+r\cdot a) &\geq -\alpha^\intercal c
    \end{align*}
    and $\alpha^\intercal a \leq 0$. 
    Equivalently, the angle between $c+\alpha$ and $p(a,b)$ is at least 90 degrees since $p(a,b)$ appears on the hemisphere opposite of $c+\alpha$.     
    Since $B(c+\alpha, r')$ intersects with $H(a,b)$, the radius $r'$ is the distance from $c+\alpha$ to $H(a,b)$ which is given by 
    \begin{align*}
        r'\geq |a^\intercal (c+\alpha) - b|& = |a^\intercal c + a^\intercal \alpha -b|\\
        &= |b-r + a^\intercal \alpha -b|\\
        &= |-r+a^\intercal \alpha|\\
        &= r + |a^\intercal \alpha|\geq r
    \end{align*}
    where the last line follows since both $-r$ and $a^\intercal \alpha$ are negative. 
    Thus, for the ball $B(c+\alpha, r')$ to intersect every hyperplane in ${\cal H}$, $r'\geq r$. 
    In both cases, $B(c,r)$ is a smallest ball intersecting every hyperplane in ${\cal H}$ thereby completing the second direction of the proof. 
\end{proof}

Next, we use this characterization to show that there exists a set of at most $k+1$ tangent hyperplanes that is sufficient to define the yolk. 
We remark that there have been previous claims that there exists exactly $k+1$ tangent hyperplanes that characterize a yolk (e.g., \cite{feld1988centripetal}).
However, the claim is not proven and, more importantly, is untrue; in the proof of Theorem \ref{thm:EvenBad} we provide an instance with an even number voters in $\mathbb{R}^2$ with fewer than $k+1=3$ tangent hyperplanes and every instance given in \cite{bailey2023dimension} has fewer than $k+1$ tangent hyperplanes even when there are an odd number of voters. 
Notably, \cite{bailey2023dimension} studies the uniqueness of the yolk solution concept and the instances used to demonstrate that the dimension of the set of yolk centers can be large have fewer than $k+1$ tangent hyperplanes. 
However, non-uniqueness is not necessary to have fewer than $k+1$ tangent hyperplanes; our instance in Theorem \ref{thm:EvenBad} has a unique yolk in $\mathbb{R}^2$ with only 2 tangent median hyperplanes. 

\begin{theorem}\label{thm:Defining}
    Given a set of ideal points $I$, let $\B(c,r)$ be a smallest ball intersecting with every median hyperplane in  a compact set ${\cal H}$. 
    There exists a subset ${\cal H}'\subseteq {\cal H}$ of size $|{\cal H}'|\leq k+1$, such that $\B(c,r)$ is a smallest ball intersecting with every median hyperplane in ${\cal H}'$.  
\end{theorem}

Prior to showing this result, we introduce notation to identify the set of hemispheres \textit{covered} by a tangent hyperplane. 
Our definition actually allows a hyperplane to cover a subset of $\mathbb{R}^k$ (Lemma \ref{lem:covering}). 
This generalization causes hyperplanes to cover convex sets which allows us to apply Helly's theorem to show that there exists a set of at most $k+1$ tangent hyperplanes that cover all hemispheres.

\begin{definition}
    Given a yolk $\B(c,r)$, let ${\cal H}(c, r)$ where ${\cal H}(c,r)=\{a\in \mathbb{R}^k:(a,a^\intercal c+r)\in {\cal H}\}$ be the set of median hyperplanes tangent to $\B(c,r)$ and let $F(a)= \{\alpha: a \cdot \alpha \ge 0\}$ for all $a\in {\cal H}(c,r)$ where $F(a)\cap \{\alpha\in \mathbb{R}^k: ||\alpha||=1\}$ is the set of hemispheres covered by the tangent hyperplane $(a,a^\intercal c+r)$. 
\end{definition}

We first show that the collection of all hyperplanes covers all of $\mathbb{R}^k$.

\begin{lemma}\label{lem:covering}
    The set of tangent hyperplanes covers all of $\mathbb{R}^k$.  
    Formally, $\bigcup_{a\in {\cal H}(c,r)} F(a)=\mathbb{R}^k$.   
\end{lemma}

\begin{proof}
    We claim for all $\alpha \in \mathbb{R}^k$, there is an $a$ such that $\alpha \in F(a)$.
    
    Case 1: If $\alpha = \Vec{0}$. 
    This case trivially holds as long as ${\cal H}(c,r)$ is non-empty since $\vec{0}\in F(a)$ for all $a$. 
    By \cite{bailey2023dimension}, the set of median hyperplanes is compact and, by Lemma \ref{lem:characterization}, there is at least one tangent median hyperplane to every hemisphere of a yolk. 
    Thus, ${\cal H}(c,r)$ is non-empty and case 1 holds. 
    
    Case 2: If $||\alpha||_2 > 0$, then let $\Bar{\alpha}=\frac{\alpha}{||\alpha||_2}$ and $||\Bar{\alpha}||_2=1$. 
    By Lemma \ref{lem:characterization}, there exists ${\cal H}(a,a^tc+r)$ that is tangent to $\He(\bar{\alpha},c,r)$. 
    By selection, ${\cal H}(a,a^tc+r) \cap \He(\bar{\alpha},c,r) = c+ar$ and since $\He(\bar{\alpha},c,r) = \{ x \in ext(\B(c, r)): \bar{\alpha}^\intercal x \ge \bar{\alpha}^\intercal c \}$, 
    
    \begin{align*}
    \Bar{\alpha}^\intercal(c+a\cdot r) &\ge  \Bar{\alpha}^\intercal c \\
    \Rightarrow\Bar{\alpha}^\intercal  \cdot a\cdot r &\ge 0\\
    \Rightarrow\Bar{\alpha}^\intercal  \cdot a &\ge 0\\
    \Rightarrow||\alpha||_2\cdot {\alpha}^\intercal \cdot a &\ge 0\\
    \Rightarrow\alpha^\intercal  \cdot a &\ge 0
    \end{align*} 
    
    and $\alpha \in F(a)$. Thus, we cover all $\alpha$-hemisphere, i.e., $\bigcup_{a\in {\cal H}(c,r)} F(a) = \mathbb{R}^k$.
\end{proof}

Thus, the condition that every $\alpha$-hemisphere of $B(c,r)$ is covered is equivalent to $\bigcup_{a\in {\cal H}(c,r)} F(a)=\mathbb{R}^k$.
Notably, by extending $F(a)$ to include all real vectors instead of just unit vectors, we convexify the set of vectors covered by a median hyperplane.  
This convexification allows us to apply Helly's theorem for a finite collection of convex sets in order to prove Theorem \ref{thm:Defining}.



\begin{lemma}(Helly's theorem)\label{lem:Helly}
    Let $C$ be a finite family of convex sets in $\mathbb{R}^k$.
    If $\bigcup_{c\in C: |c|\leq k+1} c\neq \emptyset$, then $\bigcup_{c\in C}c \neq \emptyset$.
    I.e., if the intersection of every subset of size $k+1$ is non-empty, then the intersection of the entire set is non-empty. 
\end{lemma}

\begin{proof}[Proof of Theorem \ref{thm:Defining}]
    By Lemma \ref{lem:covering}, for any unit vector $\alpha$ and any yolk $B(c,r)$, there is at least one median hyperplane tangent to the $\alpha$-hemisphere $\He(\alpha,c,r)$ and
    \begin{align*}
        &\bigcup_{a\in {\cal H}(c,r)} F(a) = \mathbb{R}^k
        \Rightarrow  \bigcap_{a\in {\cal H}(c,r)}  F^c(a) = \emptyset
    \end{align*}
    where $F^c(a) = \{\alpha: a \cdot \alpha <0\}$ is the set of $\alpha$-hemisphere $\He(\alpha,c,r)$ uncovered by the tangent hyperplane $(a,a^\intercal c + r)$.\\

    The proof of Corollary 1 in \cite{Tovey92} establishes that there exists a set ${\bar{\cal H}}$ of size $O(n^k)$ such that ${\bar{\cal H}}$ is sufficient to define the yolk, i.e., $B(c,r)$ is a smallest ball intersects with every median hyperplane in ${\bar{\cal H}}$,
    \begin{align*}
        \bigcup_{a\in \bar{\cal H}} F(a) = \mathbb{R}^k 
        \Rightarrow   \bigcap_{a\in \bar{\cal H}}  F^c(a) = \emptyset
    \end{align*}
    For contradiction, we suppose that there is no set ${\cal H}'\subseteq \bar{\cal H}\subseteq {\cal H}(c,r)$ of size $k+1$ covering all, i.e. 
    \begin{center}
    \begin{align*}
    \underset{a\in {\cal H}'}\bigcup F(a) \neq \mathbb{R}^k \Rightarrow
    \underset{a\in {\cal H}'}\bigcap F^c(a) \neq \emptyset
    \end{align*} 
    \end{center}

    
    We claim there exists a small set ${\cal H}'\subseteq \bar{\cal H}$ where $|{\cal H}'|= k+1$ and $B(c,r)$ is the smallest ball intersecting with every median hyperplane in ${\cal H}'$. We may apply Helly's theorem(Lemma \ref{lem:Helly}), since every $k+1$ sets has a nonempty intersection, then the intersection of all sets is nonempty, i.e., 
    
    \begin{align*}
    \bigcap_{a\in\bar{\cal H}} F^c(a) \neq \emptyset
    \end{align*} 
     This contradicts $\bigcap_{a\in \bar{\cal H}}  F^c(a) = \emptyset$. Thus, there exists ${\cal H}'$ of size $k+1$ can support a yolk. 
\end{proof}

\section{Upper Bound on the Ratio Between the Yolk and LP Yolk in two dimensions}\label{sec:Odd}

In this section, we establish that the bound of $1/2$ given in Theorem \ref{thm:OddR2OK} is tight when $|I|$ is odd, i.e., the LP yolk radius is always at least half the size of the yolk. 
Given an arbitrary set of ideal points, we prove this result by selecting three very specific limiting median hyperlanes and showing that the radius of a smallest ball intersecting these hyperplanes is at least one half of the yolk radius. 
The LP yolk, which covers all limiting median hyperplanes, must be large enough to cover these three hyperplanes and therefore must have at least as large of a radius. 

We begin by showing that there always exists three median hyperplanes with beneficial properties for establishing our result.  
The configuration is depicted in Figure \ref{fig:construct x-axis}.

\begin{lemma}\label{lem:generality}
    Consider an arbitrary set of ideal points $I\subset \mathbb{R}^2$ where $|I|$ is odd. 
    By shifting and rotating the coordinate axes, and re-scaling the coordinate plane, we may assume the following without loss of generality:
    \begin{enumerate}
        \item $\B(\vec{0},1)$ is a yolk.
        \item There exists a median hyperplane $\H_3$ passing through $p_3=(\cos(\eta), \sin(\eta))$ where $\eta\in [-\pi/2,0)$.
        \item There exists a median hyperplane $\H_2$ passing through $p_2=(\cos(\eta+\alpha), \sin(\eta+\alpha))$ where $\alpha\in [\pi/2,\pi/2-\eta]$.
        \item There exists a median hyperplane $\H_1$ passing through $p_1=(\cos(\eta+\alpha+\beta), \sin(\eta+\alpha+\beta))$ where $\beta\in [\pi/2,\pi]$..
        \item A limiting median hyperplane $H_{3(\nu)}$ passing through $p_3$ can be found by rotating $H_3$ clockwise by angle $\nu=\pi/2+\eta$ through $p_3$.  The resulting median hyperplane is parallel to the $x$-axis.
    \end{enumerate}
\end{lemma}

    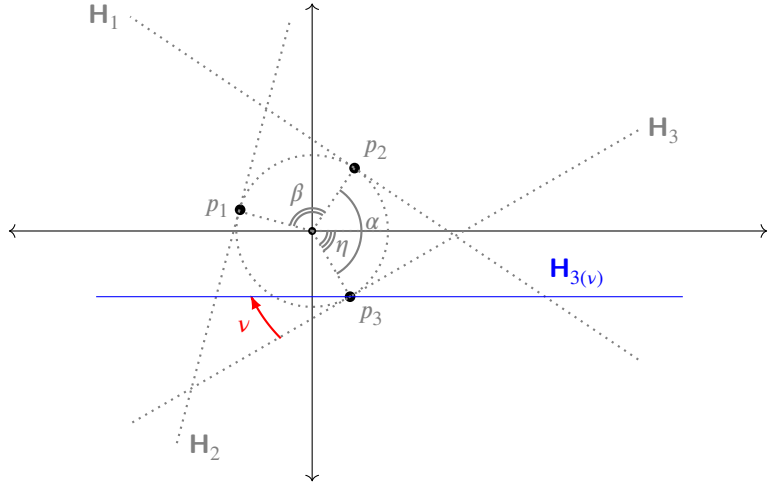
\begin{figure}[!ht]
    \centering
    \begin{tikzpicture}
    
    \def\a{56}
    \def\b{60}
    \def\x{4.3}
    \def\e{.5}
    \def\y{-2.4}
    \def\d{-30}
    \def\z{50.1}
    \def\v{-165}
    \def\g{20}
    \def\f{-1.8}
    \def\ff{-.73}
    \def\l{-.3}
    
    \pgfmathsetmacro\c{cos(\a)}
    \pgfmathsetmacro\s{sin(\a)}
    \pgfmathsetmacro\cc{cos(\a/2)/2}
    \pgfmathsetmacro\ss{sin(\a/2)/2}
    \pgfmathsetmacro\p{-\c/\s*\x+1/\s}
    \pgfmathsetmacro\xx{1/(\c)-\e}
    \pgfmathsetmacro\pp{-\c/\s*\y+1/\s}
    \pgfmathsetmacro\vv{sin(-\v)}
    \pgfmathsetmacro\uu{cos(-\v)}
    \pgfmathsetmacro\vvv{-\uu/\vv*\l+1/\vv}
    \pgfmathsetmacro\uuu{-\uu/\vv*\f+1/\vv}

    \pgfmathsetmacro\n{cos(\b)}
    \pgfmathsetmacro\m{sin(\b)}
    \pgfmathsetmacro\q{\n/\m*\x-1/\m)}
    \pgfmathsetmacro\qq{\n/\m*\y-1/\m)}
    \pgfmathsetmacro\nn{cos(\b/2)/2}
    \pgfmathsetmacro\mm{sin(\b/2)/2}
    
    \draw[black,<->] (-4,0)--(6,0);
    \draw[black, <->] (0,-3.3)--(0,3);
    
    \fill[black] (0,0) circle[radius=1.5pt];
    \fill[black] (\c,\s) circle[radius=2pt];
    \fill[black] (\n,-\m) circle[radius=2pt];
    \fill[black] (-.95,0.28) circle[radius=2pt];
    
    \node[above right] at (\c,\s) {$p_2$};
    \node[below right] at (\n,-\m) {$p_3$};
    \node[left] at (-.96,0.3) {$p_1$};
    
    
    
    \node[above right,blue] at (3,-\m) {$\H_{3(\nu)}$};

    \draw[dotted,thick] (\x,\p)--(\y,\pp);
    \draw[dotted,thick] (\x,\q)--(\y,\qq);
    \draw[dotted,thick] (\l, \vvv)--(\f,\uuu);
    \draw[dotted,thick] (0,0)--(1,0);
    
    \draw[dotted,thick] (0,0) circle[radius=1];

    \node[] at (0.8,0.1) {$\alpha$};
    \draw [thick,domain=-\b:0] plot ({.3*cos(\x)}, {.3*sin(\x)}); 
    \draw [thick,domain=-\b:0] plot ({.2*cos(\x)}, {.2*sin(\x)}); 
    \draw [thick,domain=-\b:0] plot ({.25*cos(\x)}, {.25*sin(\x)}); 
    
    \draw [thick,domain=-\b:\a] plot ({.65*cos(\x)}, {.65*sin(\x)}); 
    
    \draw [thick,domain=\a:-\v] plot ({.3*cos(\x)}, {.3*sin(\x)}); 
    \draw [thick,domain=\a:-\v] plot ({.25*cos(\x)}, {.25*sin(\x)});
    
    \node[above] at (-.2,0.2) {$\beta$};
    \node[below] at (.4,0) {$\eta$};

    \draw[dotted,thick] (\c,\s)--(0,0);
    \draw[dotted,thick] (0,0)--(\n,-\m);
    \draw[dotted,thick] (0,0)--(\uu,\vv);
    
    \coordinate (o) at (1,0);
    \coordinate (a) at (1,3);
    \coordinate (b) at (-2.1,2.76);
    
    \pgfmathsetmacro\e{(sin(\b)+sin(\a))/sin((\a)+(\b))}
    \pgfmathsetmacro\f{(cos(\b)-cos(\a))/sin((\a)+(\b))}
    \tkzDefPoint(\x,\p){A}
    \tkzDefPoint(\e,\f){B}
    \tkzDefPoint(\x,\q){C}
    
    \node[left] at (\y,\pp) {$\H_1$};
    \node[right] at (\x,\q) {$\H_3$};
    \node[right] at (-1.75,\uuu) {$\H_2$};

    \draw[-latex] [thick,domain=-135:-155,red] plot ({2*cos(\x)+1}, {2*sin(\x)});
    \node[red] at (-.9,-1.25) {$\nu$};
    
    

    \draw[rotate around={\d:(\n,-\m)},blue] (\x,\q)--(\y,\qq);
    
    \end{tikzpicture}\caption{Without loss of generality, the configuration of the three selected median hyperplanes for an arbitrary set of ideal points.}\label{fig:construct x-axis}
    \end{figure}

\begin{proof}
    The yolk and LP yolk solution concepts are defined via Euclidean measurements and therefore are preserved through shifts, re-scaling, and rotations. 
    Thus, by shifting and re-scaling, we may assume without loss of generality that $\B(\vec{0},1)$ is a yolk.
    Later on, we will only rotate the coordinate plane, and $\B(\vec{0},1)$ will remain a yolk.  
    Thus, condition (1) of the lemma holds.  

    Next, we claim this is a set of exactly 3 tangent median hyperplanes covering all hemispheres of $\B(\vec{0},1)$ --- note that Theorem \ref{thm:Defining} indicates that there is a set of at most 3. 
    By Lemma \ref{lem:characterization}, there must be a median hyperplane tangent to every hemisphere of $\B(\vec{0},1)$.  
    Therefore, there must be at least two tangent median hyperplanes. 
    If there are only two, then they must be in the form $\H(a, 1)$ and $\H(a, - 1)$ (covering opposite poles of $\B(\vec{0},1)$) implying two median hyperplanes have the same gradient. 
    However, by the constrapositive of Lemma 14 in \cite{bailey2023dimension}, if the number of ideal points is odd, then every gradient yields a unique median hyperplane, a contradiction to there being only two tangent median hyperplanes.
    Thus there are three tangent median hyperplanes that define the yolk --- i.e., that cover all hemispheres.

    Next, let $\H_1, \H_2$, and $\H_3$ be three tangent median hyperplanes that define the yolk $\B(\vec{0},1)$ as in the previous claim. 
    Let $p_i = \H_i\cap \B(\vec{0},1)$ be the point on the hyperplane that intersects with the yolks.
    We remark that $p_i$ is not necessarily an ideal point when $\H_i$ is limiting, but later show in Lemma \ref{lem:idealpoint} that $p_i$ is an ideal point if $\H_i$ is non-limiting. 
    The points $p_1,p_2$, and $p_3$ divide the ball $\B(\vec{0},1)$ into three slices with total degree $2\pi$. 
    We first claim that each slice has angle at most $\pi$; if not, then $p_1$, $p_2$ and $p_3$ appear on a same hemisphere contradicting that $\H_1$, $\H_2$, and $\H_3$ cover all hemispheres. 
    Similarly, there are at least two slices with angle at least $\pi/2$ since otherwise the third slice would have angle more than $\pi$. 
    Thus, by relabeling our points and hemispheres, we may assume $p_1$, $p_2$, and $p_3$ appear clockwise on $\B(\vec{0},1)$ where the angle between $p_1$ and $p_2$ is $\beta \geq \pi/2$ and the angle between $p_2$ and $p_3$ is $\alpha\geq \pi/2$ as depicted in Figure \ref{fig:pointsplacement}.
    Thus $\alpha \in [\pi/2,\pi]$ and $\beta\in [\pi/2,\pi]$.
    Note that if $p_3=(\cos(\eta),\sin(\eta))$ and $\eta \leq \pi/2-\alpha$ then conditions (3) and (4) of the Lemma immediately follow and it remains to show (2) and (5). 

 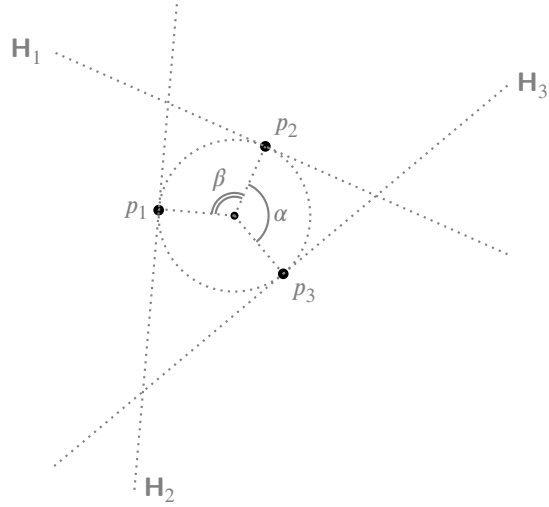
\begin{figure}[!ht]
\centering
\begin{tikzpicture}

\def\a{66}
\def\b{50}
\def\x{3.6}
\def\e{.5}
\def\y{-2.4}
\def\d{-30}
\def\z{50.1}
\def\v{-175}
\def\g{20}
\def\f{-1.32}
\def\ff{-.73}
\def\l{-.76}

\pgfmathsetmacro\c{cos(\a)}
\pgfmathsetmacro\s{sin(\a)}
\pgfmathsetmacro\cc{cos(\a/2)/2}
\pgfmathsetmacro\ss{sin(\a/2)/2}
\pgfmathsetmacro\p{-\c/\s*\x+1/\s}
\pgfmathsetmacro\xx{1/(\c)-\e}
\pgfmathsetmacro\pp{-\c/\s*\y+1/\s}
\pgfmathsetmacro\vv{sin(-\v)}
\pgfmathsetmacro\uu{cos(-\v)}
\pgfmathsetmacro\vvv{-\uu/\vv*\l+1/\vv}
\pgfmathsetmacro\uuu{-\uu/\vv*\f+1/\vv}

\pgfmathsetmacro\n{cos(\b)}
\pgfmathsetmacro\m{sin(\b)}
\pgfmathsetmacro\q{\n/\m*\x-1/\m)}
\pgfmathsetmacro\qq{\n/\m*\y-1/\m)}
\pgfmathsetmacro\nn{cos(\b/2)/2}
\pgfmathsetmacro\mm{sin(\b/2)/2}

\fill[black] (0,0) circle[radius=1.5pt];
\fill[black] (\c,\s) circle[radius=2pt];
\fill[black] (\n,-\m) circle[radius=2pt];
\fill[black] (-1,0.073) circle[radius=2pt];

\node[above right] at (\c,\s) {$p_2$};
\node[below right] at (\n,-\m) {$p_3$};
\node[left] at (-1,0.073) {$p_1$};



\node[left] at (\y,\pp) {$\H_1$};
\node[right] at (\x,\q) {$\H_3$};
\node[right] at (\f,\uuu) {$\H_2$};

\draw[dotted,thick] (\x,\p)--(\y,\pp);
\draw[dotted,thick] (\x,\q)--(\y,\qq);
\draw[dotted,thick] (\l, \vvv)--(\f,\uuu);

\draw[dotted,thick] (0,0) circle[radius=1];

\node[right] at (0.39,0) {$\alpha$};

\draw [thick,domain=\a:-\v] plot ({.3*cos(\x)}, {.3*sin(\x)}); 
\draw [thick,domain=\a:-\v] plot ({.25*cos(\x)}, {.25*sin(\x)});

\draw [thick,domain=-\b:\a] plot ({.45*cos(\x)}, {.45*sin(\x)}); 

\node[above] at (-.2,0.2) {$\beta$};

\draw[dotted,thick] (\c,\s)--(0,0);
\draw[dotted,thick] (0,0)--(\n,-\m);
\draw[dotted,thick] (0,0)--(\uu,\vv);

\coordinate (o) at (1,0);
\coordinate (a) at (1,3);
\coordinate (b) at (-2.1,2.76);

\pgfmathsetmacro\e{(sin(\b)+sin(\a))/sin((\a)+(\b))}
\pgfmathsetmacro\f{(cos(\b)-cos(\a))/sin((\a)+(\b))}
\tkzDefPoint(\x,\p){A}
\tkzDefPoint(\e,\f){B}
\tkzDefPoint(\x,\q){C}



\end{tikzpicture}\caption{Initial location of tangent hyperplanes $\H_1$, $\H_2$, and $\H_3$.}\label{fig:pointsplacement}
\end{figure}

    Prior to establishing (2) and (5), we introduce a limiting hyperplane going through $p_3$. 
    Let $\H_{3(\theta)}$ be the hyperplane obtained by rotating $\H_3$ in the clockwise direction around the point $p_3$ by angle $\theta\geq 0$. 
    Let $\nu\geq 0$ be the smallest angle where $\H_{3(\nu)}$ intersects with at least two ideal points. 
    By Lemma \ref{lem:MiddleMedian} in the appendix, $\H_{3(\theta)}$ is a median hyperplane for all $\theta \in [0,\nu]$ and $\H_{3(\nu)}$ is a limiting median hyperplane as depicted in Figure \ref{fig:search for H3}. 
    Note that $\nu=0$ if $\H_3$ is already limiting. 
    Further, $\nu < \pi-\alpha\leq \pi/2$ since otherwise $\H_{3(\pi-\alpha)}$ is a median hyperplane with the same gradient as $\H_2$, which cannot occur since we have already established that two median hyperplanes cannot have the same gradient when $|I|$ is odd in $\mathbb{R}^2$. 

    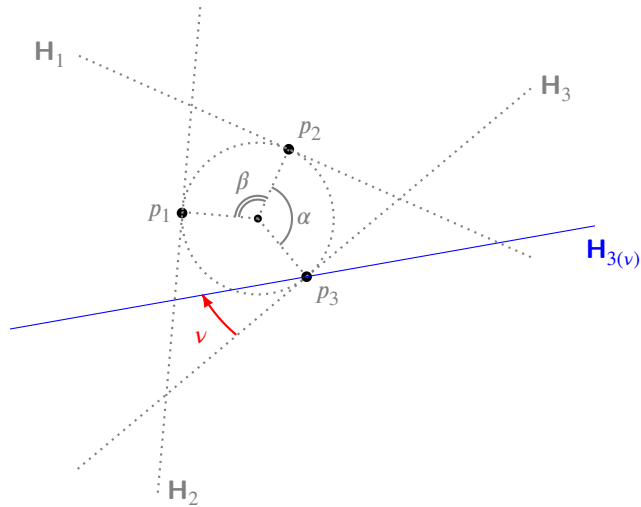
\begin{figure}[!ht]
\centering
\begin{tikzpicture}

\def\a{66}
\def\b{50}
\def\x{3.6}
\def\e{.5}
\def\y{-2.4}
\def\d{-30}
\def\z{50.1}
\def\v{-175}
\def\g{20}
\def\f{-1.32}
\def\ff{-.73}
\def\l{-.76}

\pgfmathsetmacro\c{cos(\a)}
\pgfmathsetmacro\s{sin(\a)}
\pgfmathsetmacro\cc{cos(\a/2)/2}
\pgfmathsetmacro\ss{sin(\a/2)/2}
\pgfmathsetmacro\p{-\c/\s*\x+1/\s}
\pgfmathsetmacro\xx{1/(\c)-\e}
\pgfmathsetmacro\pp{-\c/\s*\y+1/\s}
\pgfmathsetmacro\vv{sin(-\v)}
\pgfmathsetmacro\uu{cos(-\v)}
\pgfmathsetmacro\vvv{-\uu/\vv*\l+1/\vv}
\pgfmathsetmacro\uuu{-\uu/\vv*\f+1/\vv}

\pgfmathsetmacro\n{cos(\b)}
\pgfmathsetmacro\m{sin(\b)}
\pgfmathsetmacro\q{\n/\m*\x-1/\m)}
\pgfmathsetmacro\qq{\n/\m*\y-1/\m)}
\pgfmathsetmacro\nn{cos(\b/2)/2}
\pgfmathsetmacro\mm{sin(\b/2)/2}

\fill[black] (0,0) circle[radius=1.5pt];
\fill[black] (\c,\s) circle[radius=2pt];
\fill[black] (\n,-\m) circle[radius=2pt];
\fill[black] (-1,0.073) circle[radius=2pt];

\node[above right] at (\c,\s) {$p_2$};
\node[below right] at (\n,-\m) {$p_3$};
\node[left] at (-1,0.073) {$p_1$};



\node[left] at (\y,\pp) {$\H_1$};
\node[right] at (\x,\q) {$\H_3$};
\node[right] at (\f,\uuu) {$\H_2$};

\draw[dotted,thick] (\x,\p)--(\y,\pp);
\draw[dotted,thick] (\x,\q)--(\y,\qq);
\draw[dotted,thick] (\l, \vvv)--(\f,\uuu);

\draw[dotted,thick] (0,0) circle[radius=1];

\node[right] at (0.39,0) {$\alpha$};

\draw [thick,domain=\a:-\v] plot ({.3*cos(\x)}, {.3*sin(\x)}); 
\draw [thick,domain=\a:-\v] plot ({.25*cos(\x)}, {.25*sin(\x)});

\draw [thick,domain=-\b:\a] plot ({.45*cos(\x)}, {.45*sin(\x)}); 

\node[above] at (-.2,0.2) {$\beta$};

\draw[dotted,thick] (\c,\s)--(0,0);
\draw[dotted,thick] (0,0)--(\n,-\m);
\draw[dotted,thick] (0,0)--(\uu,\vv);

\coordinate (o) at (1,0);
\coordinate (a) at (1,3);
\coordinate (b) at (-2.1,2.76);

\pgfmathsetmacro\e{(sin(\b)+sin(\a))/sin((\a)+(\b))}
\pgfmathsetmacro\f{(cos(\b)-cos(\a))/sin((\a)+(\b))}
\tkzDefPoint(\x,\p){A}
\tkzDefPoint(\e,\f){B}
\tkzDefPoint(\x,\q){C}

\draw[-latex] [thick,domain=-130:-150.5,red] plot ({2*cos(\x)+1}, {2*sin(\x)});
\node[red] at (-.75,-1.56) {$\nu$};

\node[below right,blue] at (4.2,-.15) {$\H_{3(\nu)}$};

\draw[rotate around={\d:(\n,-\m)},blue] (\x,\q)--(\y,\qq);
\end{tikzpicture}\caption{We obtain the limiting median hyperplane $\H_{3(\nu)}$ by rotating $\H_3$ clockwise by the angle $\nu$. If $H_3$ is already limiting, then $\nu=0$.}\label{fig:search for H3}
\end{figure}

    To establish conditions (2) and (5), we simply rotate the coordinate plane so that (a) $H_{3(\nu)}$ is parallel to the $x$-axis, and (b) $p_3$ is on or below the $x$-axis as in Figure \ref{fig:construct x-axis}.
    If $p_3$ is on the $x$-axis, then there are two possible rotation where $p_3\in (\pm 1,0)$; in this case we select the configuration where $p_3=(1,0)$. 
    Let $\eta$ be the angle between the $x$-axis and $p_3$.  
    We define $\eta$ such that $\eta\in (-\pi,0]$ since $p_3$ is below the $x$-axis, i.e., $p_3=(\cos(\eta),\sin(\eta))$. 
    By observing the right triangle formed by $\vec{0}$, $p_3$, and the intersection of the $x$-axis with $\H_3$, we conclude $|\eta|=\pi/2-\nu$. 
    Earlier, we showed $0\leq \nu < \pi-\alpha \leq \pi/2$, and therefore $|\eta| \in (0, \pi/2]$.  
    Since $\eta<0$, $\eta \in [-\pi/2,0)$.
    Furthermore, since $|\eta|=\frac{\pi}{2}-\nu=\frac{\pi}{2}-\pi+\alpha=-\frac{\pi}{2}+\alpha$
    and $\eta<0$, $\eta\leq -\left(-\frac{\pi}{2}+\alpha\right)=\frac{\pi}{2}-\alpha$. 
    Therefore $\eta\in[-\frac{\pi}{2},\frac{\pi}{2}-\alpha)$ thereby establishing the upper bound on $\alpha$ in (3).
    We remark that this condition actually implies that (b) cannot occur, i.e,. $\eta<0$ means that $p_3$ cannot be on the $x$-axis. 
    The angle $\eta$ and $p_3$ now satisfy condition (2) and $\H_{3(\nu)}$ satisfies condition (5). 
    Thus, all 5 conditions hold. 
\end{proof}

Now that we have established these hyperplanes exist for an arbitrary instance, we show that the LP yolk radius is always at least 1/2 of the yolk radius by considering only median planes passing through $p_1$, $p_2$, and $p_3$.

\begin{theorem}\label{thm:MainHalf}
    When $|I|$ is odd in $\mathbb{R}^2$, the LP yolk radius is always at least 1/2 the size of the yolk radius.  
\end{theorem}
\begin{proof}
    Given an arbitrary instance, we will identify 3 specific limiting median hyperplanes and will show the smallest ball intersecting these three limiting median hyperplanes has radius at least 1/2 of the yolk radius.  
    Since the LP yolk must also cover these three hyperplanes, the LP yolk radius will be at least as large. 

    Without loss of generality, it suffices to consider the configuration given in the statement of Lemma \ref{lem:generality}.
    $\H_{3(\nu)}$ is our first limiting median hyperplane. 
    We find the other two by rotating $\H_2$ and $\H_1$. 
    Based on conditions $(1)-(5)$ of Lemma \ref{lem:generality}, the equations for $\H_1,\H_2, \H_3$, and $\H_{3(\nu)}$ are:
        \begin{align*}
        &\H_1=\left\{(x,y):\cos{(\alpha+\eta+\beta)}\cdot x+ {\sin{(\alpha+\eta+\beta)}}\cdot y=1\right\}\\
        &\H_2=\left\{(x,y):{\cos{(\alpha+\eta)}}\cdot x+\sin{(\alpha+\eta)}\cdot y=1\right\}\\
        &\H_3=\left\{(x,y):{\cos{(\eta)}}\cdot x+\sin{(\eta)}\cdot y =1\right\}\\
        &\H_{3(\nu)}=\left\{(x,y):y=\sin{(\eta)}\right\}\\
    \end{align*}
    Next, we identify two other limiting median hyperplanes that pass through $p_1$ and $p_2$.  
    These median hyperplanes are depicted in Figure \ref{fig:construct} and formally defined below. 
    We will use these three limiting median hyperplanes to provide a lower bound on the radius of the LP yolk.

\begin{figure}[!ht]
\centering
\begin{tikzpicture}

\def\a{56}
\def\b{60}
\def\x{3.9}
\def\i{3.3}
\def\e{.5}
\def\y{-2.2}
\def\d{-30}
\def\z{50.1}
\def\v{195}
\def\g{20}
\def\f{-1.8}
\def\ff{-.73}
\def\l{-.3}
\def\dd{-20}
\def\ddd{-80}

\pgfmathsetmacro\c{cos(\a)}
\pgfmathsetmacro\s{sin(\a)}
\pgfmathsetmacro\cc{cos(\a/2)/2}
\pgfmathsetmacro\ss{sin(\a/2)/2}
\pgfmathsetmacro\p{-\c/\s*\x+1/\s}
\pgfmathsetmacro\xx{1/(\c)-\e}
\pgfmathsetmacro\pp{-\c/\s*\y+1/\s}
\pgfmathsetmacro\vv{sin(-\v)}
\pgfmathsetmacro\uu{cos(-\v)}
\pgfmathsetmacro\vvv{-\uu/\vv*\l+1/\vv}
\pgfmathsetmacro\uuu{-\uu/\vv*\f+1/\vv}

\pgfmathsetmacro\n{cos(\b)}
\pgfmathsetmacro\m{sin(\b)}
\pgfmathsetmacro\q{\n/\m*\i-1/\m)}
\pgfmathsetmacro\qq{\n/\m*\y-1/\m)}
\pgfmathsetmacro\nn{cos(\b/2)/2}
\pgfmathsetmacro\mm{sin(\b/2)/2}

\fill[black] (0,0) circle[radius=1.5pt];
\fill[black] (\c,\s) circle[radius=2pt];
\fill[black] (\n,-\m) circle[radius=2pt];
\fill[black] (\uu,\vv) circle[radius=2pt];

\node[above right] at (\c,\s) {$p_2$};
\node[below] at (\n,-\m) {$p_3$};
\node[left] at (\uu,\vv) {$p_1$};

\node[above,blue] at (4.8,-\m) {$\H_{3(\nu)}$};
\node[left,blue] at (-.96,-3.1) {$\H_{1(\delta)}$};
\node[left,blue] at (-1.3,3.3) {$\H_{2(\gamma)}$};

\draw[dotted,thick] (\x,\p)--(\y,\pp);
\draw[dotted,thick] (\i,\q)--(\y,\qq);
\draw[dotted,thick] (\l, \vvv)--(\f,\uuu);
\draw[dotted,thick] (0,0)--(1,0);
\draw[blue] (\uu,3.6)--(\uu,-3);
\draw[dotted,thick] (0,0) circle[radius=1];



\draw[dotted,thick] (\c,\s)--(0,0);
\draw[dotted,thick] (0,0)--(\n,-\m);
\draw[dotted,thick] (0,0)--(\uu,\vv);

\coordinate (o) at (1,0);
\coordinate (a) at (1,3);
\coordinate (b) at (-2.1,2.76);

\pgfmathsetmacro\e{(sin(\b)+sin(\a))/sin((\a)+(\b))}
\pgfmathsetmacro\f{(cos(\b)-cos(\a))/sin((\a)+(\b))}
\tkzDefPoint(\uu,\vv){A}
\tkzDefPoint(\uu,3.6){B}
\tkzDefPoint(\l, \vvv){C}

\draw[-latex] [thick,domain=-130:-145,red] plot ({1.5*cos(\x)+1}, {1.5*sin(\x)});
\node[red] at (-.3,-1.15) {$\nu$};
\draw[-latex] [thick,domain=131:121,red] plot ({1.8*cos(\x)+1}, {1.8*sin(\x)});
\node[left,red] at (0,1.62) {$\gamma$};
\node[red] at (-.85,1.5) {$\delta$};
\pic [draw, -latex, thick,
      angle radius=10mm, angle eccentricity=1.2, red] {angle = C--A--B};


\draw[rotate around={\d:(\n,-\m)},blue] (\i,\q)--(\y,\qq);
\draw[rotate around={\dd:(\c,\s)},blue] (\x,\p)--(\y,\pp);

\node[] at (0.8,0.1) {$\alpha$};
\draw [thick,domain=-\b:0] plot ({.3*cos(\x)}, {.3*sin(\x)}); 
\draw [thick,domain=-\b:0] plot ({.2*cos(\x)}, {.2*sin(\x)}); 
\draw [thick,domain=-\b:0] plot ({.25*cos(\x)}, {.25*sin(\x)}); 

\draw [thick,domain=-\b:\a] plot ({.65*cos(\x)}, {.65*sin(\x)}); 

\draw [thick,domain=\a:165] plot ({.3*cos(\x)}, {.3*sin(\x)}); 
\draw [thick,domain=\a:165] plot ({.25*cos(\x)}, {.25*sin(\x)});

\node[above] at (-.2,0.2) {$\beta$};
\node[below] at (.4,0) {$\eta$};

\draw[blue] (-3.9,-\m)--(4.8,-\m);
    \draw[black,<->] (-4.5,0)--(5.4,0);
    \draw[black, <->] (0,-3.3)--(0,3.9);

\end{tikzpicture}\caption{$\H_{1(\delta)}$, $\H_{2(\gamma)}$ and $\H_{3(\nu)}$ are the three limiting median hyperplanes used to lower bound the LP yolk radius. An interactive figure that shows the resulting ``LP yolk'' for arbitrary positioning of the median and limiting median hyperplanes with the resulting LP yolk radius is located at \href{https://www.desmos.com/calculator/iiui7kanv3}{https://www.desmos.com/calculator/iiui7kanv3}.}\label{fig:construct}
\end{figure}
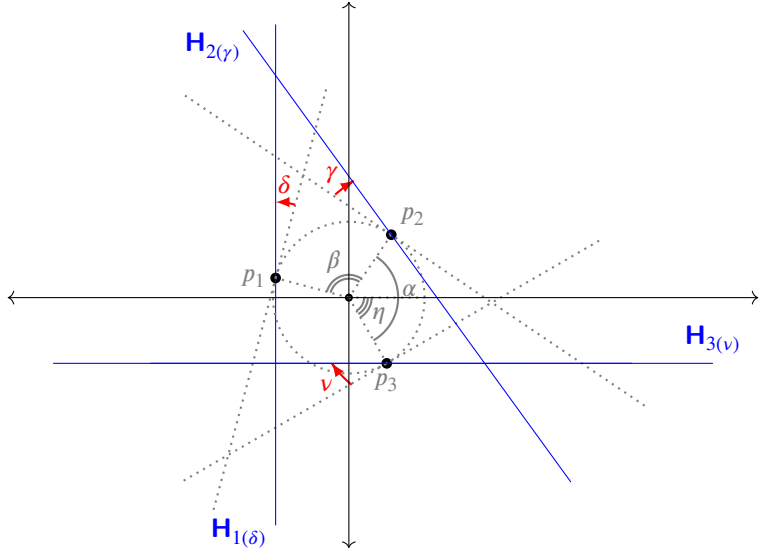

    By Lemmas \ref{lem:MiddleMedian} and \ref{lem:idealpoint} in the appendix, there is a limiting median hyperplane $H_{2(\gamma)}$ passing through $p_2$ obtained by rotating $H_2$ clockwise by angle $\gamma\geq 0$ through $p_2$. 
    The equation for $H_{2(\gamma)}$ is:
    \begin{align*}
        \H_{2(\gamma)}=\{(x,y):y=-\cot{(\alpha+\eta-\gamma)}\cdot x+\cot{(\alpha+\eta-\gamma)}\cdot \cos{(\alpha+\eta)}+\sin{(\alpha+\eta)}\}
    \end{align*}
    where $\gamma\in[0,\pi-\beta)$ since otherwise Lemma \ref{lem:MiddleMedian} implies $\H_{2(\pi-\beta)}$ is a median hyperplane with the same gradient as $\H_1$ -- 
    as established earlier, no two median hyperplanes have the same gradient when there are an odd number of ideal points.  
    
    Similarly, the limiting median hyperplane $H_{1(\delta)}$ passing through $p_1$ is obtained by rotating $H_1$ counter-clockwise by angle $\delta\geq 0$ through $p_1$.
    The equation for $H_{1(\delta)}$ is:
    \begin{align*}
        \H_{1(\delta)}=\{(x,y):y=-\cot{(\alpha+\eta+\beta+\delta)}\cdot x+\cot{(\alpha+\eta+\beta+\delta)}\cdot \cos{(\alpha+\eta+\beta)}+\sin{(\alpha+\eta+\beta)}\}
    \end{align*}
    where $\delta\in[0,\pi-\beta-\gamma)$ since otherwise $\H_{1(\pi-\beta-\gamma)}$ is a median hyperplane with the same gradient as $\H_2(\gamma)$.
    The three median limiting hyperplanes are depicted in  Figure \ref{fig:construct}.

    Since the LP yolk must cover all limiting median hyperplanes, to lower bound the size of the LP yolk, it suffices to find the smallest ball that intersects $\H_{1(\delta)}$, $\H_{2(\gamma)}$ and $\H_{3(\nu)}$.
We now characterize the smallest ball that covers all limiting median hyperplanes, and show its radius is at least $1/2$. 
This implies the ratio between a LP yolk radius and a yolk radius and  is at least 1/2 since the yolk radius is 1 by construction.

By Lemma \ref{lem:characterization}, the smallest ball covering these three hyperplanes is necessarily tangent to all three lines, i.e., the smallest ball is a circle that is inscribed in a triangle formed by these three lines. 
It is well-known, the center of a circle inscribed in a triangle is the intersection of any two interior angle bisectors. The bisector of $\H_{2(\gamma)}$ and $\H_{3(\nu)}$ is
    \begin{equation}
    \begin{aligned}
     y=\frac{-\cos{(\alpha+\eta-\gamma)}}{1+\sin{(\alpha+\eta-\gamma)}}\cdot x+\frac{\cos{(\alpha+\eta-\gamma)}\cdot\cos{(\alpha+\eta)}+\sin{(\alpha+\eta-\gamma)}\cdot\sin{(\alpha+\eta)}+\sin{(\eta)}}{1+\sin{(\alpha+\eta-\gamma)}}
    \end{aligned}
    \tag{Bisector of $\H_{2(\gamma)}$ and $\H_{3(\nu)}$}
    \end{equation}


and the bisector of $\H_{1(\delta)}$ and $\H_{3(\nu)}$ is
    \begin{equation}
    \begin{aligned}
     y=& \ \frac{-\cos{(\alpha+\eta+\beta+\delta)}}{1+\sin{(\alpha+\eta+\beta+\delta)}}\cdot x\ +\\ 
&\frac{\cos{(\alpha+\eta+\beta+\delta)}\cdot\cos{(\alpha+\eta+\beta)}+\sin{(\alpha+\eta+\beta+\delta)}\cdot\sin{(\alpha+\eta+\beta)}+\sin{(\eta)}}{1+\sin{(\alpha+\eta+\beta+\delta)}}
    \end{aligned}
    \tag{Bisector of $\H_{1(\delta)}$ and $\H_{3(\nu)}$}
    \end{equation}

Let $\B(c,r)$ be the smallest ball intersecting all three lines -- equivalently the circle inscribed in the triangle formed by the three lines. 
The center $c=(x_c,y_c)$ of the circle $\B(c,r)$ is the intersection of these two bisectors as depicted in Figure \ref{fig:3limiting}. The equation for $y_c$ is $y_c=\frac{a-b}{d}$ where
\begin{align*}
    a:=\ & \bigl(\cos{(\alpha+\eta-\gamma)}\cdot\cos{(\alpha+\eta)}+\sin{(\alpha+\eta-\gamma)}\cdot\sin{(\alpha+\eta)}+\sin{(n)}\bigr)\cdot\cos{(\alpha+\eta+\beta+\delta})\\
    b:=\ & \bigl(\cos{(\alpha+\eta+\beta+\delta)}\cdot\cos{(\alpha+\eta+\beta)}\ +\\ &\sin{(\alpha+\eta+\beta+\delta)}\cdot\sin{(\alpha+\eta+\beta)}+\sin{(n)}\bigr)\cdot\cos{(\alpha+\eta-\gamma})\\
    d:=\ & \bigl(1+\sin{(\alpha+\eta-\gamma)})\cdot\cos{(\alpha+\eta+\beta+\delta)}-(1+\sin{(\alpha+\eta+\beta+\delta)}\bigr)\cdot\cos{(\alpha+\eta-\gamma)}.
\end{align*}

Further, by (5) of Lemma \ref{lem:generality}, $\H_{3(\nu)}$ is parallel to the $x$-axis.
Since $\B(c,r)$ is tangent to this line,  $r= y_c-\sin{(\eta)}=\frac{a-b}{d}-\sin(\eta)$.
After applying a series of product-to-sum and sum-to-product trigonometric identities, $r$ simplifies\footnote{Part of the reduction was initially identified using Mathematica (\url{https://github.com/rannn-h/size-of-the-yolk}) and was verified by hand.} to 
\begin{align*}
    r=\frac{\cos{(\alpha+\eta-\gamma)}\cdot\cos{(\delta)}-\sin{(\beta+\delta+\gamma)}\cdot\sin{(\eta)}-\cos{(\alpha+\eta+\beta+\delta)}\cdot\cos{(\gamma)}}{\sin{(\beta+\delta+\gamma)}+\cos{(\alpha+\eta-\gamma)}-\cos{(\alpha+\eta+\beta+\delta)}}.
\end{align*}
The $x$-coordinate, while unneeded for the proof, of the center $c$ can be computed similarly resulting in the depiction given in Figure \ref{fig:3limiting}.
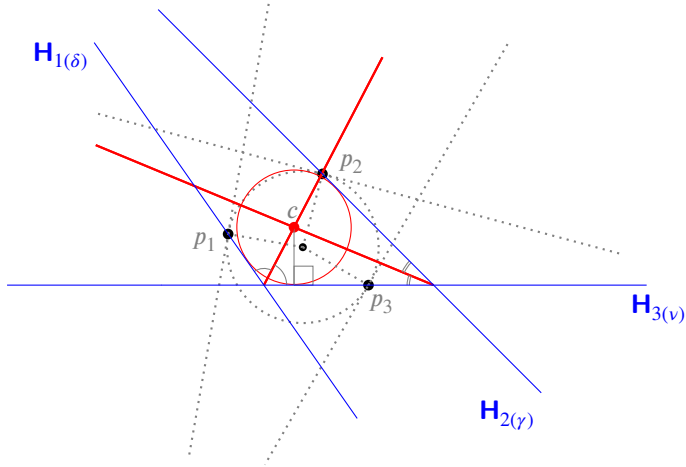
\begin{figure}[!ht]
\centering
\begin{tikzpicture}

\def\a{75}
\def\o{30}
\def\b{170}
\def\x{2.7}
\def\e{.5}
\def\y{-.5}
\def\z{-6}
\def\t{-.45}
\def\r{-1.5}
\def\d{30}
\def\f{45}
\def\g{60}
\def\h{-2.7}
\def\i{4.2}

\pgfmathsetmacro\c{cos(\a)}
\pgfmathsetmacro\s{sin(\a)}
\pgfmathsetmacro\cc{cos(\a/2)/2}
\pgfmathsetmacro\ss{sin(\a/2)/2}
\pgfmathsetmacro\p{-\c/\s*\h+1/\s)}
\pgfmathsetmacro\pp{-\c/\s*\i+1/\s}

\pgfmathsetmacro\n{cos(\b)}
\pgfmathsetmacro\m{sin(\b)}
\pgfmathsetmacro\q{-\n/\m*\t+1/\m)}
\pgfmathsetmacro\qq{-\n/\m*\r+1/\m)}
\pgfmathsetmacro\nn{cos(\b/2)/2}
\pgfmathsetmacro\mm{sin(\b/2)/2}

\pgfmathsetmacro\u{cos(\o)}
\pgfmathsetmacro\v{sin(\o)}
\pgfmathsetmacro\uu{\u/\v*\x-1/\v)}
\pgfmathsetmacro\vv{\u/\v*\y-1/\v)}


\fill[black] (\u,-\v) circle[radius=2pt];
\fill[black] (\n,\m) circle[radius=2pt];
\fill[black] (\c,\s) circle[radius=2pt];
\fill[red] (-.116,.263) circle[radius=2pt];

\fill[black] (0,0) circle[radius=1.5pt];

\draw[dotted,thick] (\h,\p)--(\i,\pp);
\draw[dotted,thick] (\t,\q)--(\r,\qq);
\draw[dotted,thick] (\x,\uu)--(\y,\vv);
\draw[dotted,thick] (0,0) circle[radius=1];

\draw[blue] (-3.9,-1/2)--(-1.86,-1/2);
\draw[thick,red] (1.725,-.5)--(-2.73,1.345);
\draw[thick,red] (-.513,-.5)--(1.05,2.503);

\draw (-.116,.263)--(-.116,-1/2);
\draw[red] (-.116,.263) circle (0.754); 
\node[above] at (-.13,.263) {$c$};


\draw[dotted,thick] (\c,\s)--(0,0);
\draw[dotted,thick] (0,0)--(\n,\m);
\draw[dotted,thick] (0,0)--(\u,-\v);

\node[below right,blue] at (4.2,-.5) {$\H_{3(\nu)}$ };
\node[above left,blue] at (-2.7,2.2) {$\H_{1(\delta)}$};
\node[below right,blue] at (2.2,-1.9) {$\H_{2(\gamma)}$};

\node[above right] at (.36,.86) {$p_2$};
\node[below right] at (.76,-\v) {$p_3$};
\node[left] at (-1,0.073) {$p_1$};

\draw[thick,red] (1.725,-.5)--(-2.73,1.345);
\draw[thick,red] (-.513,-.5)--(1.05,2.503);
\tkzDefPoint(\c,\s){A}
\tkzDefPoint(0,0){O}
\tkzDefPoint(-1,0){B}
\tkzDefPoint(\n,\m){C}
\tkzDefPoint(\u,-\v){D}
\tkzDefPoint(-.116,.263){E}
\tkzDefPoint(-.116,-1/2){F}
\tkzDefPoint(1,-1/2){G}
\tkzDefPoint(1.725,-.5){H}
\tkzDefPoint(-2.73,1.345){I}
\tkzDefPoint(-.513,-.5){J}
\tkzDefPoint(1.05,2.503){K}
\tkzMarkRightAngle(E,F,G)
\pic [draw,
      angle radius=4.2mm, angle eccentricity=1.2] {angle = A--H--E};
\pic [draw,
      angle radius=4.5mm, angle eccentricity=1.2] {angle = A--H--E};
\pic [draw,
      angle radius=3.2mm, angle eccentricity=1.2] {angle = E--H--F};
\pic [draw,
      angle radius=3.5mm, angle eccentricity=1.2] {angle = E--H--F};
\pic [draw,
      angle radius=3mm, angle eccentricity=1.2] {angle = F--J--E};
\pic [draw,
      angle radius=2.2mm, angle eccentricity=1.2] {angle = E--J--C};

\draw[rotate around={-\d:(\c,\s)},blue] (\h,\p)--(\i,\pp);
\draw[rotate around={\f:(\n,\m)},blue] (\t,\q)--(\r,\qq);
\draw[rotate around={-\g:(\u,-\v)},blue] (\x,\uu)--(\y,\vv);

\end{tikzpicture}\caption{The smallest ball intersecting the three limiting median hyperplanes $\H_{1(\delta)}, \H_{2(\gamma)}$ and $\H_{3(\nu)}$. 
The red lines are the bisector of two limiting median hyperplanes. 
The red point $c$ is the center of the smallest ball intersecting the three limiting median hyperplanes. 
The distance from $c$ to the limiting hyperplane $\H_{3(\nu)}$ is the radius of this ball.}\label{fig:3limiting}
\end{figure}


To establish $r\geq 1/2$, we first show that the middle term of numerator is non-negative, i.e., $-\sin{(\beta+\delta+\gamma)}\cdot\sin{(\eta)} \geq 0$. 
We also establish that the denominator is non-negative.  
This allows us to lower bound $r$ by removing the middle term of the numerator. 

We begin by establishing bounds on several combinations of our variables. 
\begin{itemize}
    \item $\beta+\delta+\gamma\in [\pi/2,\pi]$ since $\beta \geq 0, \delta\geq 0$ and $\delta\in [0,\pi-\beta-\delta].$
    \item $\alpha + \eta - \gamma \in [-\pi/2,\pi/2]$, since $\eta\geq -\pi/2, \alpha\in [\pi/2,\pi/2-\eta]$ and $\gamma\in [0,\pi-\beta]\subseteq [0,\pi/2]$
    \item $\alpha+\eta + \beta +\delta\in [\pi/2,3\pi/2]$ since $\eta\geq -\pi/2, \alpha\in [\pi/2,\pi/2-\eta]$, $\beta \geq \pi/2$ and $\delta \in [0,\pi-\beta-\gamma]\subset [0,\pi-\beta]$.
\end{itemize}
These inequalities respectively imply
\begin{itemize}
    \item $\sin(\beta+\delta+\gamma)\geq 0$,
    \item $\cos(\alpha+\eta-\gamma)\geq 0$,
    \item $-\cos(\alpha+\eta+\beta+\delta)\geq 0$. 
\end{itemize}
Therefore the denomiator of $r$ is non-negative and the second term of the numerator of $r$ is positive.
Thus, to lower bound $r$, we may remove the second term of the numerator, i.e., 
\begin{align*}
    r\ge f(\beta):= \frac{\cos{(\alpha+\eta-\gamma)}\cdot\cos{(\delta)}-\cos{(\alpha+\eta+\beta+\delta)}\cdot\cos{(\gamma)}}{\sin{(\beta+\delta+\gamma)}+\cos{(\alpha+\eta-\gamma)}-\cos{(\alpha+\eta+\beta+\delta)}}.
\end{align*}

Next, we will show the lower bound on $r$, $f(\beta)$, is a decreasing function of $\beta$.  
Since $\beta \leq \pi-\delta-\gamma$, to bound $r$, it then suffices to bound $f(\pi-\delta-\gamma)$.
Let $\frac{u}{v}:= \dv{f}{\beta}$ be obtained via the quotient without simplification; i.e., $v$ is the square of the denominator of $r$ implying $v\geq 0$.  
Thus, to establish $f'(\beta)\leq 0$, it suffices to show $u\leq 0$. 
By applying sum-to-product and product-to-sum identities, the numerator of $\dv{f}{\beta}$, $u$, simplifies to:
\begin{align*}
    u=&\frac{1}{2}\cos{(\alpha+\eta-\gamma)}(\cos{(\beta+\gamma)}-2\cos{(\gamma)}+\cos{(\beta+2\delta+\gamma)}-\sin{(\alpha+\eta+\beta)}\\
    &-\sin{(\alpha+\eta+\beta+2\delta)}+\sin{(\alpha+\eta+\beta+\delta-\gamma)}+\sin{(\alpha+\eta+\beta+\delta+\gamma)})
\end{align*}
    Let $A=\alpha+\eta$ where $A\in[0,\frac{\pi}{2}]$ since $\alpha\in[\frac{\pi}{2},\pi-\eta]$. 
    With this substitution, we re-write $u$ as
    \begin{align*}
        u=&\frac{1}{2}\cos{(A-\gamma)}(\cos{(\beta+\gamma)}-2\cos{(\gamma)}+\cos{(\beta+2\delta+\gamma)}-\sin{(A+\beta)}\\
    &-\sin{(A+\beta+2\delta)}+\sin{(A+\beta+\delta-\gamma)}+\sin{(A+\beta+\delta+\gamma)}). 
    \end{align*}
    Using sum-to-product identities, $u$ further simplifies to:
    \begin{align*}
        u=&\frac{1}{2}\cos{(A-\gamma)}\cdot(-2\cos{(\gamma)}+2\sin{(A+\beta+\delta)}\cdot\cos{\gamma}+2\cos{(\beta+\delta+\gamma)}\cdot\cos{(\delta)}\\
        &-2\sin{(A+\beta+\delta)}\cdot\cos{(\delta)})\\
        =&\cos{(A-\gamma)}\cdot((\sin{(A+\beta+\delta)}-1)\cdot\cos{(\gamma)}+(\cos{(\beta+\delta+\gamma)}-\sin{(A+\beta+\delta)})\cdot\cos{(\delta)}).
    \end{align*}
    Similar to before, we bound combinations of our variables: 
    \begin{itemize}
        \item $A-\gamma\in [-\pi/2,\pi/2]$ since $A\in [0,\pi/2]$ and $\gamma \in [0,\pi-\beta-\delta]\subseteq [0,\pi/2]$. 
        \item $\delta \in [0,\pi-\beta-\gamma]\in [0,\pi/2]$.
        \item $\gamma \in [0,\pi-\beta-\delta]\in [0,\pi/2]$.
    \end{itemize}
    This implies
    \begin{itemize}
        \item $\cos(A-\delta)\geq 0$.
        \item $\cos(\delta)\geq 0$.
        \item $\cos(\gamma)\geq 0$.
    \end{itemize}
    Therefore to show $u\leq 0$, it suffices to show $\cos{(\beta+\delta+\gamma)}-\sin{(A+\beta+\delta)}\leq 0$. 
    Through similar inequalities, $A+\beta+\delta \in [\pi/2,3\pi/2]$. 
    Notably, $\sin{(\cdot)}$ is decreasing on this domain.
    Since the definition of $\beta$ and $\delta$ are independent of $A$, $A$ can be made as large as possible $(\pi/2)$ while remaining in the domain $A+\beta+\delta \in [\pi/2,3\pi/2]$.
    Therefore,
    \begin{align*}
        \cos{(\beta+\delta+\gamma)}-\sin{(A+\beta+\delta)}
        \le\ & \cos{(\beta+\delta+\gamma)}-\sin{(\frac{\pi}{2}+\beta+\delta)}\\
        =\ &\cos{(\beta+\delta+\gamma)}-\cos{(\beta+\delta)}\\
        \le\ &\cos{(\beta+\delta)}-\cos{(\beta+\delta)}
        =\ 0
    \end{align*}
    Therefore $u\leq 0, \dv{f}{\beta}\leq 0$, and $r\geq f(\pi-\delta-\gamma)$.
    Finally, we establish that $f(\pi-\delta-\gamma)\geq 1/2$.
    \begin{align*}
        r&\geq f(\beta)=\frac{\cos{(\alpha+\eta-\gamma)}\cdot\cos{(\delta)}-\cos{(\alpha+\eta+\beta+\delta)}\cdot\cos{(\gamma)}}{\sin{(\beta+\delta+\gamma)}+\cos{(\alpha+\eta-\gamma)}-\cos{(\alpha+\eta+\beta+\delta)}}\\
        &\geq f(\pi-\delta-\gamma)=\frac{\cos{(\alpha+\eta-\gamma)}\cdot\cos{(\delta)}-\cos{(\alpha+\eta+\pi-\delta-\gamma+\delta)}\cdot\cos{(\gamma)}}{\sin{(\pi-\delta-\gamma+\delta+\gamma)}+\cos{(\alpha+\eta-\gamma)}-\cos{(\alpha+\eta+\pi-\delta-\gamma+\delta)}}\\
        &=\frac{\cos{(\alpha+\eta-\gamma)}\cdot\cos{(\delta)}+\cos{(\alpha+\eta-\gamma)}\cdot\cos{(\gamma)}}{\cos{(\alpha+\eta-\gamma)}+\cos{(\alpha+\eta-\gamma)}}\\
        &=\frac{\cos{(\delta)}+\cos{(\gamma)}}{2}\\
        &\geq \frac{\cos{(\delta)}+\cos{\left(\frac{\pi}{2}-\delta\right)}}{2}\\
        &=\frac{\cos{(\delta)}+\sin{(\delta)}}{2}
    \end{align*}
    since $\delta \in(0,\frac{\pi}{2}]$, and the $l^2$ norm is less than or equal to $l^ 1$ norm,
    \begin{align*}
        &\cos{(\delta)}+\sin{(\delta)}\\
        =&|\cos{(\delta)}|+|\sin{(\delta)}|\\
        =&\left\lVert(\cos{(\delta)},\sin{(\delta)})\right\rVert_1\\
        \ge&\left\lVert(\cos{(\delta)},\sin{(\delta)})\right\rVert_2\\
        =&1
    \end{align*}
    that is,
    \begin{align*}
        r&\geq \frac{\cos{(\delta)}+\sin{(\delta)}}{2}\ge \frac{1}{2}
    \end{align*}
    Thus, when $|I|$ is odd in $\mathbb{R}^2$, the yolk is at most twice the size of the LP yolk.
\end{proof}

\section{Conclusions and Future Directions}
We prove that the LP yolk radius is at least 1/2 of the yolk radius when there are an odd number of voters in $\mathbb{R}^2$. 
However, even in this simplified case, we also show that the LP yolk  can be arbitrarily far from the yolk (relative to the yolk radius). 
Further, the LP yolk radius can be arbitrarily small relative to the yolk radius in all other settings. 
For both of these negative results, our family of examples demonstrating the bound is tight appear in general position, thereby suggesting these bad examples occur with positive probability when ideal points are generated from a continuous distribution. 
Thus, the LP yolk is an unreliable, and potentially arbitrarily poor, approximation of the yolk. 

\bibliographystyle{cas-model2-names}
 \bibliography{bib}

\appendix

\section{Additional Proofs}

\begin{lemma}\label{lem:MiddleMedian}
    Suppose $I\subset \mathbb{R}^2$ and that $\H$ is a median hyperplane that goes through the ideal point $\pi_v$.
    Let $\H_\nu$ be the median hyperplane obtained by rotating $\H$ about the point $\pi_v$ in the counter-clockwise direction by angle $\nu$. 
    Let $\nu\geq 0$ be the smallest value such that $\H_{\nu}$ intersects with at least 2 ideal points, i.e., such that $\H_\nu$ is a median limiting hyperplane, then $\H_{\nu'}$ is a median hyperplane for all $\nu'\in [0,\nu]$. 
\end{lemma}

We remark that in the statement of our Lemma, that $\nu$ may be 0 if the median hyperplane is already limiting.  

\begin{proof}
    By rotating and translating the coordinate plane, we assume without loss of generality that $\H=((1,0),1)=\{(x,y):x=1\}$ and $\pi_v = (1,0)$.
    First, observe that if $\nu=0$, then the result trivially holds. 
    Thus for the remainder of the proof, we consider $\nu>0$. 

    Since $\H_{\nu}$ is the line obtained by rotating $\H$ by angle $\nu$ about the point $\pi_v$ in the counter-clockwise direction, $\H_\nu=\{(x,y):\cos(\nu)\cdot x+ \sin(\nu)\cdot y=\cos(\nu)\}$.
    We now show that if an ideal point is to the ``right'' of $\H$, then it is also to the right of $\H_{\nu'}$ for $\nu'\in [0,\nu]$. 
    Formally, we show that if $(1,0)\cdot  \pi'_v \ge 1$, then $(\cos{\nu'},\sin{\nu'})\cdot  \pi'_v\ge\cos{\nu'}$, $\forall \nu' \in [0,\nu]$. 
    We consider two cases based on whether $\pi'_v=\pi_v$. 

     Case 1: $\pi'_v=\pi_v=(1,0)$. Then $(\cos\nu,\sin\nu)\cdot \pi'_v=\cos\nu$  and $\pi'_v$ remains on the ``right'' side as desired.
    
    Case 2: $\pi'_v \neq  \pi_v = (1,0)$. Since $\nu>0$,  $\pi_v=(1,0)$ is the only ideal point on $\H=\H_0=\{(x,y):(1,0)\cdot (x,y)=1\}=\{(x,y):(\cos(0),\sin(0))\cdot(x,y)=\cos(0)\}$. Therefore $\pi'_v$ is not on this line and $(\cos{0},\sin{0})\cdot \pi'_v > \cos{0}$, i.e., it is strictly to the ``right'' of $\H$.
    For contradiction, suppose $(\cos{\nu'},\sin{\nu'})\cdot \pi_v<\cos{\nu'}$, i.e., it is to the ``left'' of $\H_\nu$. 
    Let $f(\theta)=(\cos{\theta},\sin{\theta})\cdot \pi'_v-\cos{\theta}$. 
    By definition, $f(0)=(\cos{0},\sin{0})\cdot  \pi'_v-\cos{0}>0$ and $f(\nu')<0$. 
    Further, $f(\theta)$ is a continuous 
    function, and, by the median value theorem, there must be some $\nu''\in(0,\nu')$, such that $f(\nu'')=0$, i.e. $(\cos{\nu''},\sin{\nu''})^\intercal\pi_v=\cos{\nu''}$, this means there are at least 2 ideal points on $(\cos\nu'',\sin\nu'')\cdot (x,y)=\cos\nu''$  --- both $\pi'_v$ and $\pi_v=(1,0)$. 
    This contradicts that $\nu$ is the smallest value such that at least 2 ideal points are on $\H_\nu$. 

    This completes both cases and therefore, if $(1,0)\cdot \pi'_v \ge 1$, then $(\cos{\nu'},\sin{\nu'})\cdot \pi'_v\ge \cos{(\nu')}$. 
    Thus, the number of points to the ``right'' of $\H_{\nu'}$ is at least as large as the number of points to the ``right'' of $\H$, which is at least $|I|/2$. 
    Symmetrically, the number of points to the ``left'' of $\H_{\nu'}$ is at least as large as the number of points to the ``left'' of $\H$, which is at least $|I|/2$
    and $\H_{\nu'}$ is a median hyperplane for all $\nu'\in [0,\nu]$. 
\end{proof}

\begin{lemma}\label{lem:idealpoint}
    Let $\H_i$ and $p_i$ be as in the statement of Lemma \ref{lem:generality}. 
    If $\H_i$ is a non-limiting median hyperplane, i.e., if there is only one ideal point on $\H_i$, then $p_i$ is an ideal point. 
\end{lemma}

\begin{proof}
    For contradiction, suppose $\H_i$ is a non-limiting median hyperplane, but $p_i$ is not an ideal point. That is, the ideal point $\pi$ is somewhere else on $\H_i$. 
    Without loss of generality, we assume $\B(\vec{0},1)$ is a yolk, $\H_i=((1,0),1)=\{(x,y):x=1\}$ is median hyperplane, and $p_i=\H_i\cap \B(\vec{0},1)=(1,0)$ is the point on $\H_i$ that intersects with the yolk, and $\pi=(1,\bar{y})$ is the ideal point on $\H_i$ for some $\bar{y}>0$. 
    By Lemma \ref{lem:MiddleMedian}, when we rotate $\H_i$ about the point $\pi$ in counter-clockwise direction by angle $\nu>0$ for sufficiently small $\nu$, the hyperplane $\H_{i(\nu)}$ is still a median hyperplane. 
    However, as depicted in Figure \ref{fig:idealpoint}, $\H_{i(\nu)}\cap \B(\vec{0},1)=\emptyset$, i.e., $\B(\Vec{0},1)$ is not a yolk, since a yolk is a smallest ball that intersects with all median hyperplanes. 
    This contradicts our assumption. 

    Formally, $\H_{i(\nu)}=\{(x,y): \cos(\nu)\cdot x + \sin(\nu)\cdot y = \cos(\nu) + \sin(\nu)\cdot \bar{y}\}$ and the distance from $\vec{0}$ to $\H_{i(\nu)}$ is $f(\nu):= \cos(\nu) + \sin(\nu)\cdot \bar{y}$. 
    Observe that $f'(\nu)>0$ for sufficiently small $\nu$ since $\bar{y}>0$ and that $f(0)=1$.  
    Therefore, for sufficiently small $\nu>0$, $f(\nu)> 1$ implying that $\H_{i(\nu)}\cap \B(\vec{0},1)=\emptyset$, which contradicts that $\B(\vec{0},1)$ is a yolk. 
    Thus, if $\H_i$ is non-limiting median hyperplane, then $p_i$ is an ideal point.  
\end{proof}

\begin{figure}[!ht]
\centering
\begin{tikzpicture}

\def\a{90.01}
\def\b{110}
\def\x{3}
\def\e{.5}
\def\y{-7}
\def\d{30}
\def\z{50.1}

\pgfmathsetmacro\c{cos(\a)}
\pgfmathsetmacro\s{sin(\a)}
\pgfmathsetmacro\cc{cos(\a/2)/2}
\pgfmathsetmacro\ss{sin(\a/2)/2}
\pgfmathsetmacro\p{-\c/\s*\x+1/\s)}
\pgfmathsetmacro\xx{1/(\c)-\e}
\pgfmathsetmacro\pp{-\c/\s*\y+1/\s}

\pgfmathsetmacro\n{cos(\b)}
\pgfmathsetmacro\m{sin(\b)}
\pgfmathsetmacro\q{\n/\m*\x-1/\m)}
\pgfmathsetmacro\qq{\n/\m*\y-1/\m)}
\pgfmathsetmacro\nn{cos(\b/2)/2}
\pgfmathsetmacro\mm{sin(\b/2)/2}

\fill[black] (1,1) circle[radius=3pt];
\fill[black] (0,0) circle[radius=2pt];

\node[above] at (0,0) {$\B(\Vec{0},1)$};
\node[right] at (1,0) {$p_i$};
\node[right] at (1,1) {$\pi=(1,\bar{y})$};


\node[right] at (1,3) {$\H_i$};
\node[left] at (-0.1,3) {$\H_{i(\nu)}$};


\draw[dotted,thick] (0,0) circle[radius=1];



\draw[dotted,thick] (1,0)--(0,0);

\coordinate (o) at (1,0);
\coordinate (a) at (1,3);
\coordinate (b) at (-2.1,2.76);

\pgfmathsetmacro\e{(sin(\b)+sin(\a))/sin((\a)+(\b))}
\pgfmathsetmacro\f{(cos(\b)-cos(\a))/sin((\a)+(\b))}
\tkzDefPoint(\x,\p){A}
\tkzDefPoint(\e,\f){B}
\tkzDefPoint(\x,\q){C}



\draw[-latex] [thick,domain=90:106,blue] plot ({2*cos(\x)+1}, {2*sin(\x)});
\node[blue] at (.73,2.2) {$\nu$};



\draw[dashed] (1,3)--(1,-2);
\draw[rotate around={\d:(1,1)},dashed] (1,3.3)--(1,-2.4);



\end{tikzpicture}\caption{After rotating $\H_i$ about the ideal point $\pi$ in counter-clockwise direction by angle $\nu$, $\H_{i(\nu)}\cap \B(\vec{0},1)=\emptyset$.}\label{fig:idealpoint}
\end{figure}
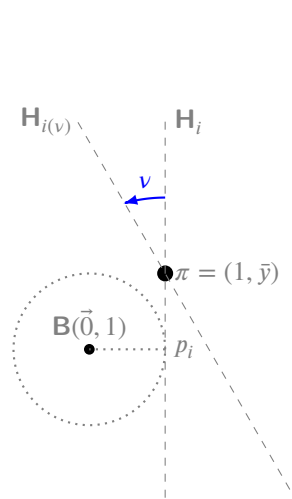

\end{document}